\documentclass{article}

\PassOptionsToPackage{numbers, compress}{natbib}
\usepackage{doi}
\usepackage[preprint]{neurips_2026}

\usepackage{geometry}
\usepackage{graphicx} 
\usepackage{amsmath}
\usepackage{bbm, bm}
\usepackage{amsthm}
\usepackage{amssymb}
\usepackage{float}
\usepackage{thmtools}
\usepackage{booktabs}

\declaretheorem[name=Proposition,numberwithin=section]{proposition}
\declaretheorem[name=Lemma,numberwithin=section]{lemma}
\declaretheorem[name=Theorem,numberwithin=section]{theorem}

\declaretheorem[name=Remark,numberwithin=section]{remark}
\declaretheorem[name=Corollary,numberwithin=section]{corollary}

\usepackage{apxproof}

\usepackage[utf8]{inputenc} 
\usepackage[T1]{fontenc}    
\usepackage{hyperref}       
\usepackage{url}            
\usepackage{amsfonts}       
\usepackage{nicefrac}       
\usepackage{microtype}      
\usepackage{xcolor}         

\title{The Dynamics of Policy Gradient\\ in Social Dilemmas with Partner Selection}

\author{
  Benedict Russell \\
  Mathematics Institute \\
  Warwick University\\
  \texttt{benedict.i.russell@warwick.ac.uk} \\
  \And
  Chin-wing Leung \\
  Department of Computer Science \\
  Warwick University \\
  \texttt{chin-wing.leung@warwick.ac.uk} \\
  \And
  Paolo Turrini \\
  Department of Computer Science\\
    Warwick University\\
  \texttt{p.turrini@warwick.ac.uk} \\
  }

\begin{document}

\maketitle

\begin{abstract}
In social dilemmas self-interested learning agents face the choice between the societal benefit of cooperation and the immediate reward of defection. Significant evidence exists on the benefits of assortment mechanisms such as partner selection for the emergence of cooperation, but this is largely available through agent-based simulations. In this paper, we provide an analytical solution to the problem, studying the policy-gradient dynamics in a multi-agent environment with partner selection. We show how partner selection changes the opponent distribution and hence the reward landscape, and prove this promotes cooperation under simple rules known from the literature. In particular, we find that population variance is a necessary condition for cooperation to emerge. Using a two-dimensional Wiener process, we extend the dynamics to capture the stochastic effects of partner selection and the resulting opponent distribution. We derive a sufficient condition for the population to be cooperation-promoting and prove the existence of a stationary distribution. Simulations confirm that the stochastic model accurately captures the policy-gradient dynamics and clarifies how the learning rate affects the emergence of cooperation.
\end{abstract}

\section{Introduction}
Social dilemmas capture the tension faced by self-interested agents when given the option to contribute to a common goal, or exploit the contributions of others. Ensuring that autonomous systems can work together for the greater good is currently one of the biggest challenges in the field of Artificial Intelligence. Recent research has explored the cooperation of LLMs \cite{willis_will_2025, nguyen_navigating_nodate}, multi-agent reinforcement learning (MARL) \cite{leung_learning_2024}, and evolutionary game theory \cite{nowak_five_2006}; each providing unique insight into how cooperation can be sustained.

In MARL, this problem is especially challenging, as agents experience a non-stationary environment induced by other learning agents \cite{anastassacos_partner_2020}. Consequently, defection is a common outcome unless the agents are equipped with an additional mechanism such as behavioural signals of the opponent \cite{priklopil_optional_2017, sabater-mir_reputation_2002} or dynamic interactions \cite{bara_enabling_2022}. In his seminal work, \citet{nowak_five_2006} identified partner selection as a key mechanism for the emergence of cooperation. This has led to extensive research on cooperation on dynamic networks \cite{bara_enabling_2022, wang_cooperation_2012, fehl_co-evolution_2011}, as well as pairwise partner selection \cite{anastassacos_partner_2020, leung_learning_2024}.

Researchers in this area commonly find that partner selection rules which increase interaction frequency or duration between cooperative agents are key to the emergence of cooperation. A popular mechanism, coined Out-for-Tat (OFT), enforces cooperative pairs to stay together; whilst re-wiring pairs with a defector \cite{izquierdo_leave_2014,  zhang_opting_2016}. Recent contributions have found allowing agents to learn partner selection is sufficient for cooperation to emerge \cite{defection_russell2026, fan_colearning}; the agents adopt OFT as a mechanism to avoid exploitation. When random matching is enforced, defection takes hold of the population.

Despite the empirical success, the current understanding of partner selection remains heavily reliant on agent-based simulations \cite{leung2026learning, leung_curiosity_2025}. While these simulations provide robust evidence that cooperation can emerge, they often fail to provide a formal theoretical justification for how these selection rules reshape the underlying learning dynamics.
Specifically, there is a lack of analytical work that maps policy gradient (PG) updates onto the non-stationary reward landscapes induced by shifting opponent distributions.
Without this theoretical foundation, it remains unclear which population conditions are strictly necessary for cooperation to take hold, or how stochasticity in action and partner selection influences long-term stability.

\paragraph{Contribution} In this paper, we provide a formal theoretical analysis of policy gradient dynamics in social dilemmas with partner selection.
Using mean-field theory, we derive the conditional partner distribution to capture how partner selection rules alter the reward structure for learning agents.
We prove that partner selection promotes cooperation under well-known rules (OFT and ROFT) and identify that population variance is a necessary condition for cooperation to emerge from an initially non-cooperative state.
We extend the mean dynamics to a stochastic model using a two-dimensional Wiener process to capture the randomness inherent in action selection and partner selection.
We derive a sufficient condition for the population to be cooperation-promoting and prove the existence of a stationary distribution under the derived model.
Simulations corroborate our theoretical findings, demonstrating that our stochastic model accurately captures the evolution of strategies and clarifies the pivotal role of the learning rate in fostering cooperative clusters.

\paragraph{Related Literature}
\textit{Partner selection.}
In evolutionary and network models, enabling cooperators to avoid defectors alters the payoff landscape, generating more assortative interaction patterns than a well-mixed population \cite{pacheco_active_2006, santos_cooperation_2006, pacheco_coevolution_2006}. This can take many forms, including active linking \cite{pacheco_active_2006, pacheco_repeated_2008, bara_enabling_2022}, partner switching \cite{fu_partner_2009, li_aspiration-based_2014}, and conditional dissociation \cite{izquierdo_leave_2014, izquierdo_option_2010}. Rewiring acts as a weak form of ostracism; defectors cannot repeatedly exploit cooperators, but cooperators have continued access to each other. These results carry over to human populations, where experimental work shows that dynamic partner updates can increase cooperation \cite{rand_dynamic_2011, fehl_co-evolution_2011, zhang_opting_2016}.
Recent research in MARL has implemented these ideas into learning environments. Anastassacos et al. \cite{anastassacos_partner_2020} show that partner selection can induce cooperation amongst selfish reinforcement learning agents. This has been extended to enable agents to learn partner selection rules themselves, with behavioural signals \cite{leung_learning_2024, leung_curiosity_2025} and without any prior information on opponents \cite{defection_russell2026}.

\textit{Learning Dynamics.}
The complexity associated with non-stationary environments has led researchers to model learning dynamics using tools from dynamical systems and evolutionary game theory. Prior work has linked learning to coupled replicator equations \cite{sato_coupled_2003, galstyan_continuous_2013}, selection-mutation models \cite{tuyls_selection-mutation_2003, bloembergen_evolutionary_2015}, and mean field approximations \cite{Hu2019ModellingTD}. Learning dynamics have been considered in population games, using mean-field theory to derive the dynamics under the concurrent learning protocol \cite{hu_dynamics_2022, Hu2019ModellingTD}. This has been extended to agents acting on graphs \cite{chu_formal_2022}, and considering stochastic interactions under the social learning protocol \cite{leung_modelling_2022}. In a closely related work, \citet{zheng_simple_2017} uses a pair approximation in an evolutionary setting to show the emergence of cooperation under a partner selection rule; our work differs by analysing a learning population.

\textit{Policy Gradient Dynamics.}
In 2-player symmetric games, \citet{srinivasan_actor-critic_nodate} provides the mean dynamics of policy gradient and compares against the replicator dynamics. \citet{bernasconi_evolutionary_2025} formally establishes the connection between the replicator dynamics and the soft-max policy gradient, and proves asymptotic stability for evolutionarily stable strategies (ESS) in symmetric games. This setting has been extended to analyse the stochastic effects of exploration \cite{leung_stochastic_evolution_2024}, with stochastic stability shown under self-play. Stochastic stability has also been analysed when the payoff faces aggregate shocks \cite{fudenberg_evolutionary_1992} and random perturbations \cite{mertikopoulos_emergence_2010}.

Whilst there has been extensive work on symmetric games, including the stochastic effects of shocks and action selection, there has not been theoretical research which shows the role of partner selection on the behaviour of reinforcement learning (RL) agents. This adds an additional layer of complexity, where the opponent distribution will influence both the mean and stochastic dynamics.

\section{Background}
In this section, we provide the necessary background on social dilemmas with partner selection and the policy gradient algorithm in repeated Prisoner's Dilemma.

\subsection{Repeated Prisoner's Dilemma}\label{section: preliminaries prinoners dilemma}
In the Prisoner’s Dilemma (PD), two players simultaneously choose whether to cooperate (C) or defect (D), receiving payoffs based on their combined actions. Whilst mutual cooperation yields the greatest collective payoff, each player is incentivised to defect and exploit their opponent. Consequently, mutual defection remains the only Nash equilibrium \cite{nash_non-cooperative_1951} and ESS in the one-shot game \cite{smith_evolution_1982}. The general payoff matrix is shown in Table \ref{table: payoff}; the payoff is such that $b>c>0$.
\begin{table}[H]
  \caption{Prisoner's Dilemma payoff matrix.}
  \label{table: payoff}
  \centering
  \begin{tabular}{lll}
    \toprule
         & C & D \\
    \midrule
    C & $b$  & $0$     \\
    D     & $b+c$ & $c$      \\
    \bottomrule
  \end{tabular}
\end{table}
This paper examines a repeated Prisoner’s Dilemma, where for each episode, an agent plays $H$ rounds of PD with their opponents. The opponent changes based on the actions of the pair, according to a predefined partner selection rule. We will consider 4 commonly studied rules \cite{leung_learning_2024,defection_russell2026}:
\begin{itemize}
    \item[--] Out-for-Tat (OFT): stay if and only if both cooperate
    \item[--] Reverse Out-for-Tat (ROFT): stay if and only if both defect
    \item[--] Always Stay (Stay): stay independent of the actions
    \item[--] Always Switch (Switch): switch independent of the actions
\end{itemize}
These provide a baseline to analyse how partner selection can be used to promote cooperation in a population of self-interested learning agents.

\subsection{Policy Gradient}\label{section: preliminaries pg}

We study independent reinforcement learners where each agent aims at optimising their expected return. The policy gradient (PG) algorithm searches over the parameter space $\bm{\psi}\in \mathbb{R}^{m}$ and updates the policy through stochastic gradient ascent to find parameters that optimise the agent's decision-making. We consider softmax parameterisation for each agent's policy, which enables balancing exploration and exploitation. The policy for action $a_k$ is given by
\begin{equation}\label{eqn: softmax function}
    \pi(a_k|\bm{\psi}(t)) := \frac{\exp(\psi_k)}{\sum_{l=1}^m\exp(\psi_{l})}.
\end{equation}
The policy update follows the REINFORCE algorithm \cite{williams_simple_1992}. If action $a_h$ is selected in round $h$ of an episode and the resulting accumulated rewards from step $h$ onwards are $R^h_{a_h}$, where we assume the discount factor is $1$, the approximated gradient of the score-function is $(R^h_{a_h}-\beta)(\bm{e}_{a_h}-\bm{\pi})$, where $\beta \in \mathbb{R}$ is a baseline \cite{Sutton2018ReinforcementL} and we set it to $0$. The parameter update across an episode of length $H$ is
\begin{equation}\label{eqn: psi reinforce update step}
    \partial_t \bm{\psi} :=\bm{\psi}(t+1)- \bm{\psi}(t) = \alpha \sum_{h=0}^{H-1}R^h_{a_h}(\bm{e}_{a_h}-\bm{\pi}).
\end{equation}
where $a_h$ is the action sampled at round $h$, and $\bm{e}_{a_h}=(0,...1,...0)$ is the basis vector. In the PD, the actions are limited to $a \in \{C, D\}$. Denoting the probability of cooperating $x := \pi(C|\bm{\psi}(t))$ and defecting $1-x := \pi(D|\bm{\psi}(t))$, the \textit{parameter} mean dynamics are given by
\begin{align*}
  \partial_t \psi_C = \alpha x(1-x)\sum_{h=0}^{H-1}(G^h_{C}-G^h_{D}), \quad\quad \partial_t \psi_D = -\partial_t \psi_C,
\end{align*}
where $G^h_C := \mathbb{E}[R^h|a_h = C]$ and $G^h_D := \mathbb{E}[R^h|a_h =D]$. Applying the chain rule to Equation \eqref{eqn: softmax function}, the \textit{policy} mean dynamics are
\begin{align}
    \frac{dx}{dt} &=
     2\alpha x^2(1-x)^2 \sum_{h=0}^{H-1}(G^h_{C}-G^h_{D}),
\end{align}
which is a one-dimensional non-linear ODE, where the sign of the episodic reward difference between cooperating and defecting will determine the evolution. Defining $G_a = \sum_{h=0}^{H-1}G_a^h$, we denote the overall difference by $\Delta G = G_C - G_D$.

\section{Model Formulation}\label{section: model}
In this section, we introduce the agent-based model and derive how partner selection alters the opponent distribution and resulting reward structure. 
\subsection{Setup}
We consider a population of agents, where the distribution of strategies is denoted by $\rho(x)$. Each strategy, $x$, corresponds to the probability of cooperating by that agent, and as such is bounded between 0 and 1. Therefore, if $x=0$ the agent always defects, whilst if $x=1$, the agent will always cooperate.

For each episode, an agent $i$ is randomly selected to be the focal agent. At round $h=0$, sample an opponent $j$ from the population; the pair interact by playing the PD game, and the payoff is given according to Table \ref{table: payoff}. The opponent $j$ is then redrawn according to the partner selection rule. This repeats for a fixed number of rounds, $H$. At the end of the episode, the focal agent $i$ updates their strategy according to the REINFORCE algorithm.

The partner selection rule chosen will play a significant role in the evolutionary dynamics. For example, under the Out-For-Tat mechanism, an agent will continue to play with their partner if and only if they both cooperate. Otherwise, a new opponent is sampled (with replacement) from the population. The distribution of these opponents is inherently different: if the opponent stayed, they are more likely to have a higher cooperation rate $x$ than a new draw from the population.

\subsection{Conditional Opponent Distribution \& Reward Structure}\label{sec: conditional opponent distribution and reward}
Consider an arbitrary agent $i$ in the population, where the strategy of this agent is denoted by $x$. At step $h$ of the episode, the distribution of the opponent given agent $i$'s strategy is $\rho^h(y|x)$. For any partner selection rule, the distribution at the next step is
\begin{align*}
    \rho^{h+1}(y|x) &= \int_0^1 \rho^{h+1}(y|x,z)\rho^{h}(z|x) dz
\end{align*}
The one-step change between the distributions is subject to the chosen Markovian partner selection mechanism. For example, under the OFT rule, the pair will stay if and only if both agents cooperate. Hence,
\begin{align*}
    \rho^{h+1}(y|x) 
    &= \int_0^1 \bigg(\delta_{z}(y)xz + \rho^{h+1}(y|x,\text{switch})(1-xz) \bigg)\rho^{h}(z|x) dz\\
    &= x\int_0^1 \delta_{z}(y)z\rho^h(z|x) dz +  \rho^{h+1}(y|x,\text{switch})\int_0^1(1-xz)\rho^{h}(z|x) dz\\
    & = xy\rho^{h}(y|x) + \rho(y)(1-xm^h(x))
\end{align*}
where, upon breaking a connection, a new opponent is sampled from the population. Given the focal agent's strategy $x$, the mean cooperation frequency  of the opponents at step $h$ is denoted by
\begin{align*}
    m^h(x) = \mathbb{E}^h[Y|x] &= \int_0^1 y \rho^h(y|x) dy.
\end{align*}
For the 4 partner selection rules, the recursion is given in Table \ref{table: iterative form and rewards for PS rules}.

The opponent distribution provides an explicit way to capture the reward at each step, $h$. We can utilise this structure to derive the conditional rewards across the episode.

\begin{table}[H]
    \caption{Conditional partner distribution update step and per round reward difference for episode length $H=2$. In each case, the initial distribution is $\rho(y)$.}
    \label{table: iterative form and rewards for PS rules}
    \centering
    \begin{tabular}{cccc}
        \toprule
         Rule &  $\rho^{h+1}(y|x)$ & $\Delta r^{h,h}(x)$ & $\Delta r^{0,1}(x)$ \\
         \midrule
         OFT&  $xy\rho^{h}(y|x) + \rho(y)(1-xm^h(x))$& $-c$ & $b(\mu_2 - \mu_1^2)$\\
         ROFT & $(1-x)(1-y)\rho^{h}(y|x) + \rho(y)(1-(1-x)(1-m^h(x)))$& $-c$ & $b(\mu_2 - \mu_1^2)$\\
         Stay & $\rho(y)$& $-c$ & $0$\\
         Switch & $\rho(y)$& $-c$ & $0$\\
         \bottomrule
    \end{tabular}
\end{table}

Following Equation \eqref{eqn: psi reinforce update step}, the episodic reward is the sum of all accumulated rewards conditioned on an action adopted at step $h=k$, followed by the original policy $x$ in future interactions. 
Consequently, we denote the agent's strategy, conditioned on its action in round $k$, as $x, C^k$ and $x, D^k$, for cooperating and defecting at $k$ respectively. Likewise, the opponent distribution, expected cooperation, and reward at time $h$ are denoted $\rho^h(y|x, C^k)$, $m_C^{k,h}(x)$, and $r_{C}^{k,h}$.

At step $h=k$, the expected difference between cooperation and defection is always $-c$. The reward difference for $h\geq k+1$ is given by 
\begin{align*}
    \Delta r^{k,h}(x) =  r_{C}^{k,h} - r_{D}^{k,h} &= \int[by+(1-x)c](\rho^h(y|x, C^k) - \rho^h(y|x, D^k))dy\\
    &= b\int y(\rho^h(y|x, C^k) - \rho^h(y|x, D^k))dy\\
    &= b (m^{k,h}_C(x) - m^{k,h}_D(x))\\
    &= b \Delta m^{k,h}(x)
\end{align*}
where given an agent has policy $x$ and cooperated in the $k$th round,  $m^{k,h}_C(x)$ is explicitly given by
\begin{align}\label{eqn: conditional mean step h}
    m^{k,h}_C(x) = \mathbb{E}^h[Y|x,C^k] &= \int_0^1 y \rho^h(y|x,C^k) dy.
\end{align}
Let $\mu_l=\int y^l\rho(y)dy$ denote the $l$th moment of the underlying population distribution $\rho(y)$; the reward differences for the first two steps in the episode are shown in Table \ref{table: iterative form and rewards for PS rules} (a full derivation is provided in Appendix \ref{appendix: derivations}).

Regardless of the partner selection rule, the expected reward difference between cooperation and defection in the $k$th round is $-c$. Therefore, for cooperation to emerge the future rewards must compensate by pairing the focal agent with more cooperative agents which can build sustained partnerships. Interestingly, both OFT and ROFT provided the same expected reward across the first two rounds: keeping defective pairs together is as effective as keeping cooperative ones. We will now provide a more general result for when $H>2$.

The conditional reward difference for Always stay and Always switch is zero for all steps $h\geq k+1$. Here, we aim to show that the reward difference for all time steps $h \geq k+1$ is non-negative for the OFT and ROFT mechanisms. 

\begin{proposition}\label{prop: OFT and ROFT have r^h_k geq 0}
For the OFT and ROFT partner selection rules, the reward difference $\Delta r^{k,h}(x) \geq 0$ for all $x \in [0,1]$ and $h \geq k+1$. 
\end{proposition}
\begin{proof}
   All proofs are provided in Appendix \ref{appendix: proofs}. 
\end{proof}

\begin{corollary}\label{cor: H=2 is sufficient}
   Under the OFT and ROFT partner selection rules, if $\Delta G[\rho] >0$ for $H=2$ then $\Delta G[\rho] >0$ for all $H\geq 2$.
\end{corollary}

The above result allows a much simpler reward structure to be studied, associated with the partner selection mechanisms. In particular, increasing the episodic length can only increase the marginal accumulated rewards for cooperation. Therefore, the results in Section \ref{section: macroscopic model} provide sufficient conditions for the emergence of cooperation for all $H\geq 2$.

\subsection{Evolution of the Mean Dynamics}\label{section: macroscopic model}
In this section, we show how the above reward structure can be embedded within the policy gradient update. We then prove that always stay and always switch converge to the pure defection equilibrium, regardless of the initial distribution. Finally, we show that both OFT and R-OFT necessarily increase cooperation when the initial population has sufficient variance.

Under the random matching and always stay mechanisms, the reward difference simplifies to $-Hc$ for any episode length $H$. As such, when $H=2$ the policy update is given by
\begin{align*}
    \frac{dx}{dt} &= -4\alpha cx^2(1-x)^2
\end{align*}
which generates the mean-field continuity PDE
\begin{align*}\label{eqn: mean pde for stay/switch}
    \partial_t \rho + \partial_ x(-4\alpha x^2(1-x)^2 \rho) &=0.
\end{align*}
\begin{theorem}\label{thm: convergence to pure defection}
    Let $\rho_0 \in L^1(0,1), \rho_0\geq 0$, and $\int_0^1\rho_0 = 1$. Under the Stay and Switch rules, the population converges to pure defection. 
\end{theorem}
Consequently, we know that always-switch and always-stay will, under the mean-dynamics, cause mass defection to take over the population. Therefore, we seek to show that partner selection rules such as OFT and R-OFT can induce a more cooperative population.
Under the OFT mechanism, the continuity equation is
\begin{align*}
\partial_t \rho + \partial_ x(u[\rho]( x) \rho) &=0, & u[\rho]( x)=2\alpha\Delta G[\rho] x^2(1- x)^2,
\end{align*}
where for $H=2$, we have
\begin{align*}
    \Delta G[\rho] & = \sum_{h=0}^{H-1} \Delta G^h[\rho]\\
    &= \Delta r^{0,0}(x) + \Delta r^{0,1}(x) + \Delta r^{1,1}(x)\\
    &= -c + b(\mu_2-\mu_1^2) -c \\
    &= b\text{Var}(\rho(y)) -2c.
\end{align*}
By Corollary \ref{cor: H=2 is sufficient}, for $H>2$ this value is a lower bound for the sign of the drift. As a consequence, if cooperation emerges for $H=2$, it provides a sufficient condition for longer episode lengths. Denote $\mathcal{P}([0,1])$ as the set of all probability measures on the domain.

\begin{proposition}\label{prop: existence to continuity equation}
    Let $\rho_0 \in \mathcal{P}([0,1])$, then there exists a unique solution to the continuity equation. Furthermore, the unique solution can be expressed as 
    \begin{align*}
        \rho(t) = (X_{t})_\# \rho_0
    \end{align*}
    where $(X_t)_\# \rho_0$ is the push-forward of $\rho_0$ by the characteristic flow.
\end{proposition}

Define the time integral of the non-local velocity induced by the episodic reward,
\begin{align*}
    K(t) = \int_0^t \Delta G[\rho(s)]\; ds
\end{align*}
Then the continuity PDE depends on time through $K$, and hence satisfies
\begin{equation}
    \rho(t,\cdot) = (X_{K(t)})_\#\rho_0.
\end{equation}
See Appendix \ref{appendix: proofs} for details. This enables the flow of the mass to be uniquely determined by the behaviour of $K$.
\begin{theorem}\label{thm: convergence to cooperation}
    Let $\rho_0 \in L^1(0,1), \rho_0\geq 0$, and $\int_0^1\rho_0 = 1$. If $\Delta G[\rho_0]>0$, then $K(t) \nearrow K^*$ where $K^* \in \mathbb{R}$ and $\rho(t) \rightharpoonup (X_{K^*})_\#\rho_0$. That is, if the initial population variance is high enough, the limiting distribution is shifted towards higher cooperation.
\end{theorem}

We highlight that a finite $K^*$ is attained, which means $X_{K^*}(x) < 1$. This enforces that the population does not converge to a single strategy, and hence does not converge to pure cooperation. The positivity of the reward difference $\Delta G[\rho_0]$ constrains the payoff; the variance on the domain $[0,1]$ is bounded above by $\frac{1}{4}$, meaning $b>4c$ is a necessary condition when $H=2$. This constraint relaxes as $H$ increases, as the long term benefits of OFT and ROFT increase the reward difference.

\section{Stochastic Dynamics}
To capture the stochasticity induced by action exploration and the opponent distribution, we model the episodic REINFORCE update as a random increment in the parameter space, $\bm{\psi}$. To do this we follow the approach in \cite{kifer_random_1988, leung_stochastic_evolution_2024}, characterising the random increment by the first two moments. In particular, we approximate the discrete stochastic updates by a continuous-time diffusion process:
\begin{align}\label{eqn: parameter SDE}
    d\bm{\psi} = \bm{\mu} dt + \sqrt{\Sigma}\;d\bm{W}_t,
\end{align}
where $\bm{\mu}$ is the expected update, $\bm{W}_t$ is a two-dimensional standard Wiener process, and $\Sigma:=\Sigma(\bm{\psi},t)$ denotes the covariance matrix. In 2-action repeated games, we can exploit the logit difference to simplify the stochastic dynamics. Applying Ito's lemma \cite{ito_stochastic_1944}, we derive the policy evolution as
\begin{align}\label{eqn: probability of cooperating stochastic update}
    dx = [2\alpha x^2(1-x)^2\Delta G[\rho]  + 2x(1-x)(1-2x)\Sigma_{CC}]\;dt + 2x(1-x)\sqrt{\Sigma_{CC}}\;dW_t 
\end{align}
where $\Sigma_{CC}$ denotes the variance of the parameter update. Denote the $h$th update step of the REINFORCE algorithm as $U^h := (R_{a_h}^h -\beta)(1\{a_h = C\}-x)$. The covariance is 
\begin{align*}
     \Sigma_{CC} &:= \alpha^2\sum_{h=0}^{H-1} \Big(x(1-x)^2S^h_C + x^2(1-x)S^h_D - x^2(1-x)^2(\Delta G^h[\rho])^2\Big) + 2\alpha^2\sum_{h<l} \text{Cov}(U^h,U^l),
\end{align*}
where $S^h_C  = \mathbb{E}[(R^h-\beta)^2|a_h=C]$ and $S^h_D  = \mathbb{E}[(R^h-\beta)^2|a_h=D]$ are the second moments of the reward conditioned on actions $C$ and $D$ at round $h$. For any action $a$, this is given by
\begin{align*}
    S^h_a &= 
    \sum_{l=h}^{H-1}\text{Var}(r_l|a_h=a) + 2 \sum_{h\leq l<k}\text{Cov}(r_l,r_k|a_h=a) + (G^h_a-\beta)^2.
\end{align*}
Details of the variance and covariance terms are provided in Appendix \ref{appendix: derivations}.

\subsection{Population Evolution}
By propagating the stochastic updates across the population, we can study the strategy evolution at a population scale. The Fokker-Planck equation (FPE) describes the evolution of the probability density function as a diffusion process \cite{risken_fokker-planck_1996}. For the strategy space $x$, the FPE corresponding to \eqref{eqn: probability of cooperating stochastic update} is
\begin{align}\label{eqn: fokker plank}
    \partial_t\rho(t,x) = -\partial_x[A_\rho(t,x)\rho(t,x)] +\frac{1}{2}\partial_{xx}[B^2_\rho(t,x)\rho(t,x)]
\end{align}
where
\begin{align*}
    A_\rho(t,x) &= 2\alpha x^2(1-x)^2 \Delta G[\rho] + 2x(1-x)(1-2x)\Sigma_{CC},\quad B^2_\rho(t,x) = 4x^2(1-x)^2\Sigma_{CC}.
\end{align*}
Given an initial strategy distribution $\rho(0,x)=\rho_0(x)$, the time evolution can be solved numerically. 

These numerical solutions will enable us to analyse the long-term dynamics; for the short-term we can provide an equivalent sufficient condition to the mean dynamics for an increasing level of cooperation. 

In particular, the next result shows that given a sufficiently low learning rate, the cooperation levels in the population will initially increase for sufficiently high variance in the initial population.
\begin{proposition}\label{prop: increase cooperation for some finite time}
    Suppose $\rho_0$ has non-zero interior mass and $\Delta G[\rho_0] >0$. Then there exists a $T\in \mathbb{R}^+$  and $\alpha^*$ such that for all $\alpha < \alpha^*$, the mean level of cooperation is increasing on $[0, T]$.
\end{proposition}

It does not, however, guarantee that a stable cooperative equilibrium will be reached. When $H=2$, the population variance remaining sufficiently high would prove a sufficient condition, but the evolution of the variance depends upon higher moments. To analyse the equilibrium which is attained in the stochastic case, we prove the limiting distribution is well-posed and turn to numerical solutions to analyse the behaviour.

\subsection{Stationary Distribution}
The stationary distribution occurs at the point such that the time derivative in the PDE is zero. Applying no-flux boundary conditions, this condition then reduces to finding a solution to the ODE given by
\begin{align}\label{eqn: steady state ODE}
    \partial_x [B_\rho^2(x)\rho(x)] &= 2A_\rho(x)\rho(x).
\end{align}
A major obstacle in proving such an equation is well-defined is the behaviour at the boundary. At $x=0$ and $x=1$, the noise term $B^2(x)=0$ which can prevent using standard techniques. We first consider an $\varepsilon-$regularised version, then extend the existence as $\varepsilon\rightarrow0$. Fix $\eta \in L^\infty([0,1])$ as a trial density such that $A_\eta(x)$ and $B^2_\eta(x)$ are a function of $x$ and $\eta$ only. The regularised versions are $A_\eta^\varepsilon(x) = A_\eta(x)$ and $B^{2,\varepsilon}_\eta(x) = B^2_\eta(x)+\varepsilon$. Define the set 
\begin{align*}
    S^\varepsilon =\{\eta \in C([0,1]): \eta \geq 0, \int_0^1\eta(x) dx = 1, \|\eta\|_\infty \leq M^\varepsilon\}
\end{align*}
For a given $\eta \in S$, the unique solution $\mathcal{F}^\varepsilon[\eta](x)$ to the linear ODE is given by
\begin{align*}
   \mathcal{F}^\varepsilon[\eta](x) &= \frac{w^\varepsilon_\eta(x)}{\int_0^1 w^\varepsilon_\eta(x)dx},\qquad w^\varepsilon_\eta(x) = \frac{1}{B^{2,\varepsilon}_\eta(x)}\exp(\int_0^x\frac{2A^\varepsilon_\eta(y)}{B_\eta^{2,\varepsilon}(y)}dy)
\end{align*}
By updating the value of $\eta$ using the operator $\mathcal{F}$, we aim to converge to a fixed point of the system, which will correspond to a stationary distribution of the PDE in \eqref{eqn: fokker plank}. The required regularity of the operator and space is proven in Appendix \ref{appendix: proofs}.
\begin{proposition}\label{prop: existence to regularised problem}
    Fix $\varepsilon>0$, then there exists at least one fixed point of the mapping $\mathcal{F}^\varepsilon[\rho] = \rho$. That is, there is at least one solution to the regularised steady-state equations.
\end{proposition}

\begin{remark}
    One can consider a PDE with additional noise not captured by the first two moments as an independent stochastic fluctuation of size $\sigma$. The corresponding stationary distribution would be equivalent to the regularisation above, and therefore existence would follow. Moreover, uniqueness would follow for sufficiently high $\sigma$ by analysing the Lipschitz constant of the operator, $\mathcal{F}$. 
\end{remark}
We can now extend this to find an unregularised solution by sending the parameter $\varepsilon \rightarrow 0$. 
\begin{theorem}\label{thm: existence to un-regularised problem}
    There exists at least one stationary probability measure $\rho \in \mathcal{P}([0,1])$ which solves \eqref{eqn: steady state ODE} in the weak sense.
\end{theorem}

The result implies that the steady state solutions can contain Dirac measures at the boundary. Consequently, we would expect in the long run for mass to accumulate in clusters of pure defection and cooperation. Moreover, the steady state solution is non-unique, explaining how the dynamics are strongly influenced by the initial population and parameter regimes.

\section{Experiments}
In this section, we show how the theoretical analysis and model capture the dynamics displayed in the agent-based simulations. To do so, we solve Equation \eqref{eqn: fokker plank} for the 4 different partner selection rules. 

\subsection{Game Setting}
We use a finite volume method over the domain $[0,1]$ to capture the time and spatial evolution. The agent-based simulations act as a ground truth in the experiments; we conduct 30 simulations of a population of 1,000 agents training over $E=$5e6 episodes. The population distribution at time $t$ is obtained by averaging the simulations. For comparison, the simulation time when plotting is scaled such that $t = E/N$.

The initial strategy values are sampled from a specified distribution $\rho_0$, with the parameter $\bm{\psi}(0)$ obtained through the inverse of the softmax function \eqref{eqn: softmax function}. The payoff is given by Table \ref{table: payoff}, where $b,c = 3,0.1$ respectively. Unless specified, the parameters are $\alpha = 0.01, H=2$. Evolution under a wide range of initial conditions is shown in Appendix \ref{appendix: empirical study}.

\subsection{Result}
We compare the distribution of strategies in the population over time. Figure \ref{fig: comparison of evolutions for 4 PS rules} presents the results of both the theoretical solution (solid line) and the simulations (histogram). We clearly see that the FPE derived captures the distribution of strategies as it evolves through time. This close alignment is present for all 4 partner selection rules, where the behaviour of the two cooperation-inducing rules (OFT and ROFT) displays very similar dynamics which lead to the emergence of a defecting and cooperating cluster. On the other hand, the rules Stay and Switch induce defection to take over the population. 
\begin{figure}[h]
     \centering
    \includegraphics[width=0.8\textwidth]{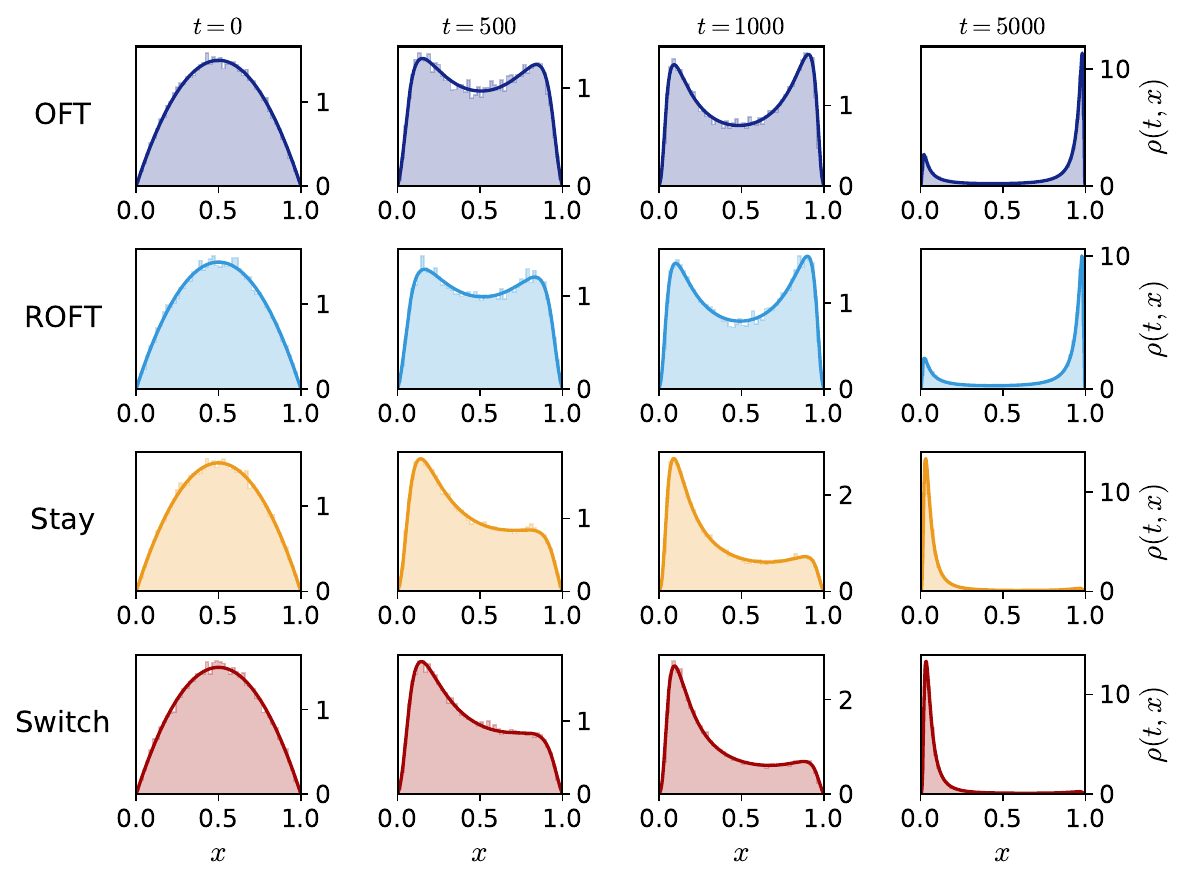}
     \caption{Evolution of the strategy distribution where the population is initialised with $X\sim\beta(2,2)$ for the 4 rules. The theoretical solution (solid line) matches the simulations (histogram).}
     \label{fig: comparison of evolutions for 4 PS rules}
\end{figure}

As highlighted in Section \ref{section: macroscopic model}, the underlying distribution requires sufficient variance for cooperation to emerge under the mean dynamics. This begs the question of how and why cooperation can emerge when the population is initialised to a singular policy \cite{leung_learning_2024, defection_russell2026}. The answer lies in how the learning rate can induce population variance. With low learning rates, the stochastic effects are diminished and therefore policies are updated close to the underlying expectation. A higher value of $\alpha$ enables sampled trajectories to influence the dynamics and therefore induce a faster-growing population variance. Figure \ref{fig: mean evolution with dirac initial condition} shows how for an initially unbiased population, a cooperation cluster can be supported through only increasing the learning rate. Whilst higher learning rate can induce greater variance and support cooperation, the FPE approximation loses fidelity.

\begin{figure}[H]
     \centering
    \includegraphics[width=0.82\textwidth]{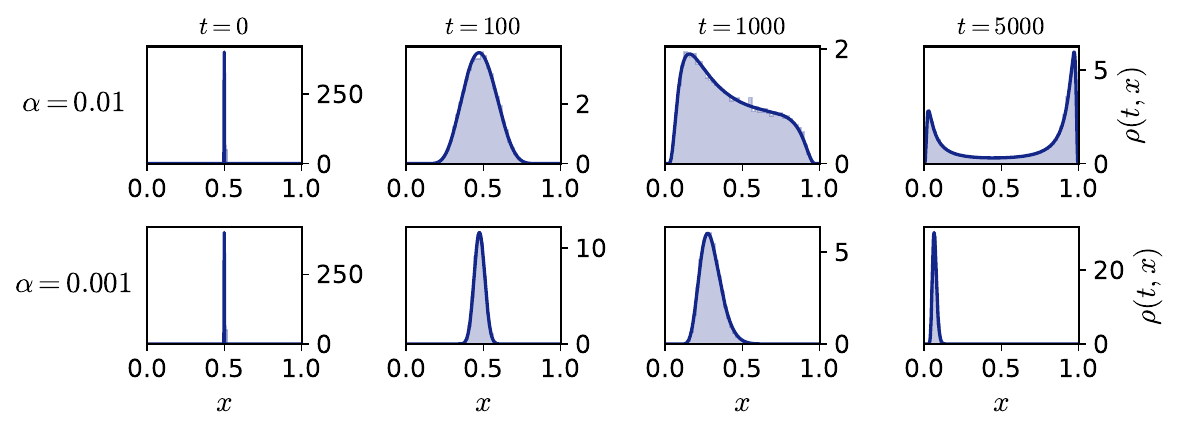}
     \caption{Evolution of the strategy distribution under OFT where the population is initialised with a Dirac at 0.5. The effect of the learning rate is clear: higher values induce additional variance in the population which in turn induces cooperation. Time has been scaled with the learning rate for comparison.}
     \label{fig: mean evolution with dirac initial condition}
\end{figure}

\section{Conclusion}
In this paper, we study the learning dynamics of policy gradient in a multi-agent environment when partner selection rules are enforced. We show how the opponent distribution alters the reward landscape, and prove that the underlying dynamics promote cooperation for both the OFT and ROFT partner selection rules. In particular, we find that population variance is a fundamental requirement for the emergence of cooperation under partner selection. We extend the model to capture the stochastic dynamics using a 2-dimensional Wiener process which models the stochasticity induced by action and opponent selection. In this setting, we show how the same variance requirement is sufficient for an initial increase in cooperation with a sufficiently small learning rate, before proving the existence of a stationary distribution. This analysis revealed that the structure of the long-term dynamics constitutes mass at the pure defection and pure cooperation strategies. In the experiments, we demonstrate that the stochastic model accurately describes the policy gradient dynamics and can be used to understand the role of the learning rate in inducing cooperation. Future directions include extending the model to other learning algorithms and considering a synchronous population partner selection model.

\newpage
\bibliography{refs}
\newpage

\appendix
\section{Derivations}\label{appendix: derivations}

\subsection{Conditional Partner Distribution and Rewards}
For additional clarity, we provide the explicit iterative form for the first two steps of the episode with the 4 partner selection rules. All steps are computed using the iterative formulas provided in Table \ref{table: iterative form and rewards for PS rules}. For each, the initial opponent distribution is set as $\rho(y)$.
\begin{table}[H]
    \caption{Conditional partner distribution and reward difference for OFT}
    \label{table: condition partner distribution OFT}
    \centering
    \begin{tabular}{cccccc}
    \toprule
         $h$ & $\rho^h(y|x, C^0)$ & $\rho^h(y|x, D^0)$  & $m^{0,h}_C(x)$ & $m^{0,h}_D(x)$ & $\Delta r^{0,h}(x)$ \\
         \midrule
         0 & $\rho(y)$ & $\rho(y)$ & $\mu_1$ & $\mu_1$ & $-c$\\
         1 & $y\rho(y) + (1-\mu_1)\rho(y)$ & $\rho(y)$ &$\mu_2-\mu_1^2 +\mu_1$ & $\mu_1$ &b$(\mu_2-\mu_1^2)$ \\
         \bottomrule
    \end{tabular}
\end{table}


\begin{table}[H]
    \caption{Conditional partner distribution and reward difference for ROFT}
    \label{table: condition partner distribution ROFT}
    \centering
    \begin{tabular}{cccccc}
    \toprule
         $h$ & $\rho^h(y|x, C^0)$ & $\rho^h(y|x, D^0)$  & $m^{0,h}_C(x)$ & $m^{0,h}_D(x)$ & $\Delta r^{0,h}(x)$ \\
         \midrule
         0 & $\rho(y)$ & $\rho(y)$ & $\mu_1$ & $\mu_1$ & $-c$\\
         1 & $\rho(y)$ & $ (1+\mu_1-y)\rho(y)$ &$\mu_1$ & $\mu_1 + \mu_1^2-\mu_2 $ &b$(\mu_2-\mu_1^2)$ \\
    \bottomrule
    \end{tabular}
\end{table}
 
\begin{table}[H]
    \caption{Conditional partner distribution and reward difference for Stay and Switch}
    \label{table: condition partner distribution Stay / switch}
    \centering
    \begin{tabular}{cccccc}
    \toprule
         $h$ & $\rho^h(y|x, C^0)$ & $\rho^h(y|x, D^0)$  & $m^{0,h}_C(x)$ & $m^{0,h}_D(x)$ & $\Delta r^{0,h}(x)$ \\
         \midrule
         0 & $\rho(y)$ & $\rho(y)$ & $\mu_1$ & $\mu_1$ & $-c$\\
         $\geq 1$ & $\rho(y)$ & $\rho(y)$ & $\mu_1$ & $\mu_1$ & $0$\\
    \bottomrule
    \end{tabular}
\end{table}
When $H=2$, the only remaining reward difference to calculate is $\Delta r^{1,1}(x) = -c$ for any partner selection rule. For larger values of $H$, the rewards depend on higher-order moments and can be calculated using the iterative formulas. Taking the sum across these rewards gives $\Delta G$ for each rule. 

\subsection{Stochastic Dynamics in 2-action Symmetric Games}
Recall the update to the parameter is
\begin{align*}
    \Delta \psi_C = \alpha \sum_{h=0}^{H-1} U^h,
\end{align*}
where $U^h= (R^h_{a_h}-\beta)(1\{a_h=C\} -x)$.
Define the mean vector
\begin{align*}
    \mu(\bm{\psi},t) := \mathbb{E}[\Delta\bm{\psi}(t)],
\end{align*}
and the covariance matrix by
\begin{align*}
    \Sigma(\bm{\psi},t) &= \mathbb{E}[(\Delta\bm{\psi}(t) - \mu)(\Delta\bm{\psi}(t)-\mu)^T].
\end{align*}
In the two-action case, since $\Delta\psi_C = -\Delta\psi_D$, we have
\begin{align*}
\mu = (\mu_C,-\mu_C), \qquad 
    \Sigma = \begin{pmatrix}
        \Sigma_{CC} & -\Sigma_{CC}\\
        -\Sigma_{CC} & \Sigma_{CC}
    \end{pmatrix}
\end{align*}
where $\mu_C = \mathbb{E}[\Delta\psi_C] = \alpha x(1-x)\Delta G[\rho]$, and $\Delta G[\rho] = G_C - G_D$. To compute the covariance, note that 
\begin{align*}
    \Sigma_{CC} &= \text{Var}(\Delta\psi_C) = \alpha^2 \sum_{h=0}^{H-1} \text{Var}(U^h) + 2 \alpha^2 \sum_{0\leq h <k \leq H-1} \text{Cov}(U^h,U^k).
\end{align*}
For each round $h$, the variance term is
\begin{align*}
    \text{Var}(U^h) &= x(1-x)^2 S^h_C + x^2(1-x) S^h_D - x^2(1-x)^2 (\Delta G^h)^2
\end{align*}
where $S^h_a = \mathbb{E}[(R^h-\beta)^2|a_h=a]$ and $\Delta G^h = G^h_C - G^h_D$. For $h<k$, define the conditional moment
\begin{align*}
    M^{h,k}_{a,b} &= \mathbb{E}[(R^h-\beta)(R^k-\beta)|a_h = a, a_k =b],
\end{align*}
then each covariance term is 
\begin{align*}
    \text{Cov}(U^h,U^k) &= x^2(1-x)^2(M^{h,k}_{C,C} -M^{h,k}_{D,C}-M^{h,k}_{C,D}+M^{h,k}_{D,D} -\Delta G^h \Delta G^k).
\end{align*}

We then approximate the discrete stochastic updates with the continuous time diffusion process as defined in Equation \eqref{eqn: parameter SDE}. To simplify the equations, consider the logit difference given by $z =\psi_C-\psi_D$. Then $\Delta z =\Delta \psi_C - \Delta \psi_D = 2\Delta \psi_C$. The corresponding one-dimensional parameter dynamics are
\begin{align}
    dz = \mu_z dt + \sigma_z dW_t
\end{align}
where 
\begin{align*}
    \mu_z = 2\alpha x(1-x)\Delta G[\rho],\quad    \sigma_z^2 = 4 \Sigma_{CC}.
\end{align*}
The cooperation probability of an agent is then computed with the softmax function \begin{align*}
    x(z) = \frac{1}{1+e^{-z}}.
\end{align*}
Applying Ito's lemma with 
\begin{align*}
    x'(z) = x(1-x), \qquad x''(z) = x(1-x)(1-2x)
\end{align*}
we obtain
\begin{align*}
    dx &= [x(1-x)\mu_z + \frac{1}{2}x(1-x)(1-2x)\sigma^2_z]\; dt + x(1-x) \sigma_z\; dW_t. 
\end{align*}
Substituting in the formulas of $\mu_z$ and $\sigma_z$ gives the provided equation.
\subsection{Conditional Moments}
For action $a \in\{C,D\}$, the second moment of the episodic reward is
\begin{align*}
    S^h_a &= \mathbb{E}[(R^h-\beta)^2|a_h=a]\\
    &= \text{Var}(R^h-\beta|a) + \mathbb{E}[R^h-\beta|a_h=a]^2\\
    &= \text{Var}(R^h|a_h=a) + (G^h_a-\beta)^2\\
    &=\sum_{l=h}^{H-1}\text{Var}(r_l|a_h=a) + 2 \sum_{h\leq l<k}\text{Cov}(r_l,r_k|a_h=a) + (G^h_a-\beta)^2.
\end{align*}
Let the opponent type at timestep $l$ be given by $Y_l$. For the focal agent, the action at the conditioned round is fixed:
\begin{align*}
    z_h = 1\{a_h = C\},
\end{align*}
and for all round $l>h$, the policy is given by $x$. Then the actions sampled in each round $l>h $ are 
\begin{align*}
    \zeta_l\sim \text{Ber}(Y_l),\quad  z_l \sim \text{Ber}(x)
\end{align*}
for the opponent and focal agent, respectively. Note that the focal agent's actions are i.i.d. across the episode. The reward at round $l$ is given by 
\begin{align*}
    r_l = b\zeta_l + c(1-z_l)
\end{align*}
At the initial round $l=h$, the action of the focal agent is fixed, so 
\begin{align*}
    \mathbb{E}[r_h|a_h=a] &= b\mathbb{E}[\zeta_h |a_h=a] + c \cdot 1\{a_h = D\}\\
    &= b m^h(x) + c \cdot 1\{a_h = D\}.
\end{align*}
For all $l> h$, the conditional mean for the focal agent taking action $a$ at time $h$ is
\begin{align*}
    \mathbb{E}[\zeta_l|a_h=a] = \mathbb{E}[Y_l|a_h=a]  = m_a^{h,l}(x) 
\end{align*}
which means the expected conditional reward is 
\begin{align*}
    \mathbb{E}[r_l|a_h=a] &= b \mathbb{E}[\zeta_l|a_h=a] + c\mathbb{E}[1-z_l]\\
    &= b m_a^{h,l}(x)  +c(1-x)
\end{align*}
Likewise, the conditional variance when $l=h$ is
\begin{align*}
    \text{Var}(r_h|a_h=a) &= b^2 \text{Var}(\zeta_h|a_h=a)\\
     &= b^2m^h(x)(1-m^h(x)),
\end{align*}
and for $l>h$ it is
\begin{align}
     \text{Var}(r_l|a_h=a) &= b^2 \text{Var}(\zeta_l|a_h=a)  + c^2\text{Var}(1-z_l) \nonumber \\
     &= b^2m_a^{h,l}(x)  (1-m_a^{h,l}(x) ) + c^2 x(1-x).\label{eqn: conditional variance on a}
\end{align}
The conditional means can be computed using repeated substitution of the iterative partner distribution update. For the covariance terms, let $h\leq l<k\leq H-1$. Then
\begin{align*}
    \text{Cov}(r_l,r_k|a_h=a) & = \text{Cov}(b\zeta_l+(1-z_l)c, b\zeta_k+c(1-z_k)|a_h=a)\\
    &=b^2\text{Cov}(\zeta_l,\zeta_k|a_h=a) + bc\text{Cov}(\zeta_l,1-z_k |a_h=a) + bc\text{Cov}(\zeta_k,1-z_l |a_h=a)\\ &\quad + c^2\text{Cov}(1-z_l,1-z_k|a_h=a)\\
    &=b^2\text{Cov}(Y_l,Y_k|a_h=a) + bc\text{Cov}(Y_l,1-z_k |a_h=a) + bc\text{Cov}(Y_k,1-z_l |a_h=a)
\end{align*}
where the final term in the second line is zero due to the independence of action selection by the focal agent. Note that for $H=2$, the final two terms are zero as $l=h$ is the only relevant term and the focal action is fixed. In this case, this covariance is given by
\begin{align*}
    \text{Cov}(Y_l,Y_k|a_h=a) &=\mathbb E[Y_lY_k|a_h=a] - m_a^{h,l}(x) m_a^{h,k}(x) .
\end{align*}
\subsubsection*{Computing the Covariance for $H=2$}
Beginning with the variance terms, we have for $h=0$:
\begin{align*}
    \mathbb{E}[r_0|Y_0,a_0 = C] &= b\mathbb{E}[\zeta_0|Y_0] = bY_0\\
    \text{Var}(r_0|Y_0,a_0 = C) &= b^2\text{Var}(\zeta_0|Y_0) = b^2Y_0(1-Y_0)\\
    \Rightarrow \text{Var}(r_0|a_0 = C) &= \mathbb{E}[b^2Y_0(1-Y_0)] + \text{Var}( bY_0) \\
    &=b^2(\mu_1-\mu_1^2).
\end{align*}
Likewise for defection:
\begin{align*}
    \mathbb{E}[r_0|Y_0,a_0 = D] &= b\mathbb{E}[\zeta_0|Y_0] + c = bY_0+c\\
    \text{Var}(r_0|Y_0,a_0 = D) &= b^2\text{Var}(\zeta_0|Y_0) = b^2Y_0(1-Y_0)\\
    \Rightarrow \text{Var}(r_0|a_0 = D) &= \mathbb{E}[b^2Y_0(1-Y_0)] + \text{Var}(bY_0+c) \\
    &=b^2(\mu_1-\mu_1^2).
\end{align*}
For the following step: $h=1$,
\begin{align*}
    \mathbb{E}[r_1|Y_1,a_0 = C] &= b\mathbb{E}[\zeta_1|Y_1] + c\mathbb{E}[1-z_1]= bY_1 + c(1-x)\\
    \text{Var}(r_1|Y_1,a_0 = C) &= b^2\text{Var}(\zeta_1|Y_1) + c^2\text{Var}(1-z_1) = b^2Y_1(1-Y_1) + c^2x(1-x)\\
    \Rightarrow \text{Var}(r_1|a_0 = C) &= \mathbb{E}[b^2Y_1(1-Y_1) + c^2x(1-x)|a_0 = C] + \text{Var}(bY_1 + c(1-x)|a_0 = C) \\
    &=b^2\Big(\mathbb{E}[Y_1|a_0 = C] - \mathbb{E}[Y_1|a_0 = C]^2\Big)+c^2x(1-x).
\end{align*}
Note this is the same form as Equation \eqref{eqn: conditional variance on a}. Using the one-step Markov operator on the opponent distribution, we can derive the reward variances for each of the partner selection rules.

The updated covariances can be computed using the conditional moment
\begin{align*}
    M^{0,1}_{a,b} &= \mathbb{E}[(R^0-\beta)(R^1-\beta)|a_0 = a, a_1 = b]\\
    &= \mathbb{E}[(r_0 + r_1-\beta)(r_1-\beta)|a_0 = a, a_1 = b]\\
    &= \mathbb{E}[r_0r_1|a_0 = a, a_1 = b] + \mathbb{E}[r_1^2|a_0 = a, a_1 = b] - \beta \mathbb{E}[r_0|a_0 = a, a_1 = b] - 2\beta \mathbb{E}[r_1|a_0 = a, a_1 = b] +\beta^2.
\end{align*}
\paragraph{OFT}
The expected cooperation rate of the opponent in the next step given the current opponent and the focal agent's action is
\begin{align*}
    \mathbb{E}[Y_1|Y_0,a_0 = C] = Y_0\cdot Y_0 + (1-Y_0) \mu_1.
\end{align*}
Then, computing the conditional mean we have
\begin{align*}
    \mathbb{E}[Y_1|a_0 = C] &= \mathbb{E}[\mathbb{E}[Y_1|Y_0,a_0 = C]]\\
    &= \mathbb{E}[Y_0^2 + (1-Y_0)\mu_1]\\
    &= \mu_2 + \mu_1 -\mu_1^2\\
     \mathbb{E}[Y_1|a_0 = D] &= \mu_1 
\end{align*}
For the covariance terms, we have
\begin{align*}
    \mathbb{E}[Y_0 Y_1|a_0 = C] &= \mathbb{E}[Y_0\mathbb{E}[Y_1|Y_0,a_0 = C]]\\
    &= \mathbb{E}[Y_0(Y_0^2 + (1-Y_0)\mu_1)]\\
    &= \mu_3 + \mu_1^2 - \mu_2\mu_1\\
    \mathbb{E}[Y_0 Y_1|a_0 = D] &= \mathbb{E}[Y_0\mathbb{E}[Y_1|Y_0,a_0 = D]]\\
    &= \mathbb{E}[Y_0\mu_1]\\
    &= \mu_1^2
\end{align*}
The variance terms are
\begin{align*}
    \text{Var}(r_1|a_0 = C) = b^2(\mu_2+\mu_1-\mu_1^2)(1-\mu_2-\mu_1+\mu_1^2) + c^2x(1-x),\;\; \text{Var}(r_1|a_0 = D) = b^2\mu_1(1-\mu_1) + c^2x(1-x).
\end{align*}
The covariance terms are
\begin{align*}
    \text{Cov}(r_0,r_1| a_0 = C) &= b^2\text{Cov}(Y_0,Y_1|a_0 = C)\\
    &= b^2(\mathbb{E}[Y_0 Y_1|a_0 = C] - \mathbb{E}[Y_0|a_0 = C]\mathbb{E}[Y_1|a_0 = C])\\
    &=b^2(\mu_3 + \mu_1^2 - \mu_2\mu_1 - \mu_1(\mu_2+\mu_1-\mu_1^2))\\
    &= b^2(\mu_3 -2\mu_2\mu_1 + \mu_1^3)\\
    \text{Cov}(r_0,r_1| a_0 = D) &= b^2\text{Cov}(Y_0,Y_1|a_0 = D)\\
    &= b^2(\mathbb{E}[Y_0 Y_1|a_0 = D] - \mathbb{E}[Y_0|a_0 = D]\mathbb{E}[Y_1|a_0 = D])\\
    &=b^2(\mu_1^2 - \mu_1^2)= 0
\end{align*}
It remains to condition on the second step, $h=1$. Explicitly, the terms are given by
\begin{align*}
    \mathbb{E}[r_1|a_1=C] &= bm^1(x)\\
        &= b\int_0^1 y[xy\rho(y) + \rho(y)(1-x\mu_1)]dy\\
        &= b(\mu_1+x(\mu_2-\mu_1^2))\\
    \mathbb{E}[r_1|a_1=D] &= bm^1(x)+c\\
        &= b(\mu_1+x(\mu_2-\mu_1^2))+c\\
    \text{Var}(r_1|a_1=a) &= b^2m^1(x)(1-m^1(x))\\
        &= b^2(\mu_1+x(\mu_2-\mu_1^2))(1-(\mu_1+x(\mu_2-\mu_1^2)))
\end{align*}
Collating these terms, we can summarise the elements of the second moment in Table \ref{table: second moment oft}.

\begin{table}[H]
    \caption{Second moments ($S^h_a$) of episodic reward when $H=2$ under OFT.}
    \label{table: second moment oft}
    \centering
    \begin{tabular}{ccp{20em}c}
    \toprule
         $h$ & $a$ & $\text{Var}(R^h|a_h=a)$ & $(G^h_a-\beta)^2$\\
         \midrule
         $0$ & $C$ & $b^2[(\mu_2+\mu_1-\mu_1^2)(1-\mu_2-\mu_1+\mu_1^2)+(\mu_1-\mu_1^2)\newline +2(\mu_3 -2\mu_2\mu_1 + \mu_1^3)]+c^2x(1-x)$ & $[b(\mu_2+2\mu_1-\mu_1^2)+c(1-x)-\beta]^2$\\
         & $D$ & $2b^2\mu_1(1-\mu_1) + c^2x(1-x)$ & $[2b\mu_1 + c(2-x)-\beta]^2$ \\
      \midrule
         $1$ & $C$ & $b^2(\mu_1+x(\mu_2-\mu_1^2))(1-\mu_1-x(\mu_2-\mu_1^2))$ & $[b(\mu_1+x(\mu_2-\mu_1^2))-\beta]^2$\\
         & $D$ & $b^2(\mu_1+x(\mu_2-\mu_1^2))(1-\mu_1-x(\mu_2-\mu_1^2))$ & $[b(\mu_1+x(\mu_2-\mu_1^2)) +c-\beta]^2$ \\
    \bottomrule
    \end{tabular}
\end{table}

\paragraph{ROFT}
The expected cooperation rate of the opponent in the next step given the current opponent and the focal agent's action is
\begin{align*}
    \mathbb{E}[Y_1|Y_0,a_0 = C] &= \mu_1\\
    \mathbb{E}[Y_1|Y_0,a_0 = D] &= Y_0\cdot(1-Y_0) + Y_0 \mu_1   
\end{align*}
Then computing the conditional mean,
\begin{align*}
    \mathbb{E}[Y_1|a_0 = C] &= \mu_1\\
     \mathbb{E}[Y_1|a_0 = D] &=  \mathbb{E}[\mathbb{E}[Y_1|Y_0,a_0 = D]]\\
    &= \mathbb{E}[Y_0(1-Y_0) + Y_0 \mu_1]\\
    &= \mu_1 +\mu_1^2 - \mu_2
\end{align*}
For the covariance terms, we have
\begin{align*}
    \mathbb{E}[Y_0 Y_1|a_0 = C] &= \mathbb{E}[Y_0\mathbb{E}[Y_1|Y_0,a_0 = C]]\\
    &= \mathbb{E}[Y_0\mu_1]\\
    &= \mu_1^2\\
    \mathbb{E}[Y_0 Y_1|a_0 = D] &= \mathbb{E}[Y_0\mathbb{E}[Y_1|Y_0,a_0 = D]]\\
    &= \mathbb{E}[Y_0(Y_0(1-Y_0) + Y_0 \mu_1)]\\
    &= \mu_2 - \mu_3 + \mu_1\mu_2
\end{align*}
The variance terms are then:
\begin{align*}
    \text{Var}(r_1|a_0=C) = b^2\mu_1(1-\mu_1) + c^2x(1-x),\;\; \text{Var}(r_1|a_0=D) = b^2(\mu_1 +\mu_1^2 - \mu_2)(1-\mu_1 -\mu_1^2 + \mu_2) + c^2x(1-x)
\end{align*}
The covariance terms are
\begin{align*}
    \text{Cov}(r_0,r_1| a_0 = C) &= b^2\text{Cov}(Y_0,Y_1|a_0 = C)\\
    &= b^2(\mathbb{E}[Y_0 Y_1|a_0 = C] - \mathbb{E}[Y_0|a_0 = C]\mathbb{E}[Y_1|a_0 = C])\\
    &=b^2(\mu_1^2- \mu_1^2) = 0\\
    \text{Cov}(r_0,r_1| a_0 = D) &= b^2\text{Cov}(Y_0,Y_1|a_0 = D)\\
    &= b^2(\mathbb{E}[Y_0 Y_1|a_0 = D] - \mathbb{E}[Y_0|a_0 = D]\mathbb{E}[Y_1|D])\\
    &=b^2(\mu_2-\mu_3+\mu_1\mu_2 -\mu_1(\mu_1+\mu_1^2-\mu_2))\\
    &=b^2(\mu_2-\mu_1^2+2\mu_1\mu_2 -\mu_1^3-\mu_3)
\end{align*}

It remains to condition on the second step, $h=1$. Explicitly, the terms are given by
\begin{align*}
    \mathbb{E}[r_1|a_1=C] &= bm^1(x)\\
        &= b\int_0^1 y[(1-x)(1-y)\rho(y) + \rho(y)(1-(1-x)(1-\mu_1))]dy\\
        &= b(\mu_1-(1-x)(\mu_2-\mu_1^2))\\
    \mathbb{E}[r_1|a_1=D] &= bm^1(x)+c\\
        &= b(\mu_1-(1-x)(\mu_2-\mu_1^2))+c\\
    \text{Var}(r_1|a_1=a) &= b^2m^1(x)(1-m^1(x))\\
        &= b^2(\mu_1-(1-x)(\mu_2-\mu_1^2))(1-\mu_1+(1-x)(\mu_2-\mu_1^2)))
\end{align*}
Collating these terms, we can summarise the elements of the second moment in Table \ref{table: second moment roft}.

\begin{table}[H]
    \caption{Second moments ($S^h_a$) of episodic reward when $H=2$ under ROFT.}
    \label{table: second moment roft}
    \centering
    \begin{tabular}{ccp{25em}c}
    \toprule
         $h$ & $a$ & $\text{Var}(R^h|a_h=a)$ & $(G^h_a-\beta)^2$\\
         \midrule
         $0$ & $C$ & $2b^2\mu_1(1-\mu_1) + c^2x(1-x)$ & $[2b\mu_1 + c(1-x)-\beta]^2$\\
         &$D$ & $b^2[(\mu_1 +\mu_1^2 - \mu_2)(1-\mu_1 -\mu_1^2 + \mu_2)+\mu_1(1-\mu_1) \newline + 2(\mu_2-\mu_1^2+2\mu_1\mu_2 -\mu_1^3-\mu_3)]  + c^2x(1-x)$ & $[b(2\mu_1+\mu_1^2-\mu_2) +c(2-x)-\beta]^2$ \\
         \midrule
         $1$ & $C$ & $b^2(\mu_1-(1-x)(\mu_2-\mu_1^2))(1-\mu_1+(1-x)(\mu_2-\mu_1^2))$ & $[b(\mu_1-(1-x)(\mu_2-\mu_1^2))-\beta]^2$\\
         &$D$ & $b^2(\mu_1-(1-x)(\mu_2-\mu_1^2))(1-\mu_1+(1-x)(\mu_2-\mu_1^2))$ & $[b(\mu_1-(1-x)(\mu_2-\mu_1^2))+c-\beta]^2$ \\
         \bottomrule
    \end{tabular}
\end{table}

\paragraph{Always Stay}
Since the transition is independent of the action, for both actions $a \in \{C, D\}$ the expected cooperation rate of the opponent in the next step given the current opponent and action is
\begin{align*}
    \mathbb{E}[Y_1|Y_0,a_0=a] = Y_0.
\end{align*}
Then computing the conditional mean,
\begin{align*}
    \mathbb{E}[Y_1|a_0=a] &= \mathbb{E}[\mathbb{E}[Y_1|Y_0,a_0=a]]=\mathbb{E}[Y_0]=\mu_1
\end{align*}
For the covariance terms, we have
\begin{align*}
    \mathbb{E}[Y_0 Y_1|a_0=a] &= \mathbb{E}[Y_0\mathbb{E}[Y_1|Y_0,a_0=a]]\\
    &= \mathbb{E}[Y_0Y_0]\\
    &= \mu_2
\end{align*}
The variance terms are then:
\begin{align*}
    \text{Var}(r_1|a_0=a) = b^2\mu_1(1-\mu_1)+c^2x(1-x)
\end{align*}
The covariance terms are
\begin{align*}
    \text{Cov}(r_0,r_1| a_0=a) &= b^2\text{Cov}(Y_0,Y_1|a_0=a)\\
    &= b^2(\mathbb{E}[Y_0 Y_1|a_0=a] - \mathbb{E}[Y_0|a]\mathbb{E}[Y_1|a_0=a])\\
    &=b^2(\mu_2 - \mu_1^2)\\
    &=b^2 \text{Var}(\rho)
\end{align*}

It remains to condition on the second step, $h=1$. Explicitly, the terms are given by
\begin{align*}
    \mathbb{E}[r_1|a_1=C] &= bm^1(x)\\
        &= b\int_0^1 y\rho(y)dy\\
        &= b\mu_1\\
    \mathbb{E}[r_1|a_1=D] &= bm^1(x)+c\\
        &= b\mu_1+c\\
    \text{Var}(r_1|a_1=a) &= b^2m^1(x)(1-m^1(x))\\
        &= b^2\mu_1(1-\mu_1)
\end{align*}
Collating these terms, we can summarise the elements of the second moment in Table \ref{table: second moment stay}.

\begin{table}[H]
    \caption{Second moments ($S^h_a$) of episodic reward when $H=2$ under Always Stay.}
    \label{table: second moment stay}
    \centering
    \begin{tabular}{cccc}
    \toprule
         $h$& $a$ & $\text{Var}(R^h|a_h=a)$ & $(G^h_a-\beta)^2$\\
         \midrule
         $0$ & $C$ & $2b^2(\mu_1 +\mu_2-2\mu_1^2) + c^2x(1-x)$ & $[2b\mu_1 + c(1-x)-\beta]^2$\\
         & $D$ & $2b^2(\mu_1 +\mu_2-2\mu_1^2) + c^2x(1-x)$ & $[2b\mu_1 + c(2-x)-\beta]^2$ \\
        \midrule
          $1$ & $C$ & $b^2\mu_1(1-\mu_1)$ & $[b\mu_1 -\beta]^2$\\
         & $D$ & $b^2\mu_1(1-\mu_1)$ & $[b\mu_1+c-\beta]^2$ \\
         \bottomrule
    \end{tabular}
\end{table}

\paragraph{Always Switch}
Since the transition is independent of the action, for both actions $a \in \{C, D\}$ the expected cooperation rate of the opponent in the next step given the current opponent and action is
\begin{align*}
    \mathbb{E}[Y_1|Y_0,a_0=a] = \mu_1.
\end{align*}
Then computing the conditional mean,
\begin{align*}
    \mathbb{E}[Y_1|a] &= \mathbb{E}[\mathbb{E}[Y_1|Y_0,a_0=a]]=\mathbb{E}[\mu_1]=\mu_1
\end{align*}
For the covariance terms, we have
\begin{align*}
    \mathbb{E}[Y_0 Y_1|a_0=a] &= \mathbb{E}[Y_0\mathbb{E}[Y_1|Y_0,a_0=a]]\\
    &= \mathbb{E}[Y_0\mu_1]\\
    &= \mu_1^2
\end{align*}
The variance terms are then:
\begin{align*}
    \text{Var}(r_1|a_0=a) = b^2\mu_1(1-\mu_1)+c^2x(1-x)
\end{align*}
The covariance terms are
\begin{align*}
    \text{Cov}(r_0,r_1| a_0=a) &= b^2\text{Cov}(Y_0,Y_1|a_0=a)\\
    &= b^2(\mathbb{E}[Y_0 Y_1|a_0=a] - \mathbb{E}[Y_0|a_0=a]\mathbb{E}[Y_1|a_0=a])\\
    &=b^2(\mu_1^2 - \mu_1^2)= 0
\end{align*}
Note that for conditioning at step $h=1$, the same derivation as Always Stay holds. Collating these terms, we can summarise the elements of the second moment in Table \ref{table: second moment switch}.
\begin{table}[H]
    \caption{Second moments ($S^h_a$) of episodic reward when $H=2$ under Always Switch.}
    \label{table: second moment switch}
    \centering
    \begin{tabular}{cccc}
    \toprule
         $h$ & $a$ & $\text{Var}(R^h|a_h=a)$ & $(G^h_a-\beta)^2$\\
         \midrule
         $0$ & $C$ & $2b^2\mu_1(1-\mu_1) + c^2x(1-x)$ & $[2b\mu_1 + c(1-x)-\beta]^2$\\
         &$D$ & $2b^2\mu_1(1-\mu_1) + c^2x(1-x)$ & $[2b\mu_1 + c(2-x)-\beta]^2$ \\
         \midrule
          $1$ & $C$ & $b^2\mu_1(1-\mu_1)$ & $[b\mu_1 -\beta]^2$\\
         & $D$ & $b^2\mu_1(1-\mu_1)$ & $[b\mu_1+c-\beta]^2$ \\
         \bottomrule
    \end{tabular}
\end{table}

\section{Missing Proofs}\label{appendix: proofs}

\subsection{Conditional Rewards}
We begin by analysing the base case of $k=0$. Defining $\Delta \rho^{0,h+1}(y) = \rho^{h+1}(y|x,C^0) -\rho^{h+1}(y|x,D^0)$, the system reduces to analysing the following recursion:
\begin{align*}
    \Delta \rho^{0,h+1}(y) &= xy\Delta\rho^{0,h}(y) - x\rho(y)\Delta m^{0,h}\\
    \Delta m^{0,h} &= \int_0^1 y\Delta\rho^{0,h}(y)dy
\end{align*}
Define the operator $\mathcal{T}$ as the recursive update on a function $g$. Formally,
\begin{align*}
    \mathcal{T}g := Yg- \mathbb{E}[Yg]
\end{align*}
then $g^{h+1} =\mathcal{T}g^h$. For $h>1$
\begin{align}\label{eqn: g iterative definition}
    g^h(y) &= yg^{h-1}(y) - \mathbb{E}[Yg^{h-1}(Y)],\quad    g^1(y) = y-\mu_1.
\end{align}
Then we can isolate the partner dependence $y$, and from the policy $x$.
\begin{lemma}\label{lemma: recursion simplification oft, k=0}
    Let $\rho(y)$ be the initial distribution, and $g^h$ be defined by Equation \eqref{eqn: g iterative definition}. Then for all $h\geq1$, the difference in distribution and mean of the opponent is given by
    \begin{align*}
    \Delta \rho^{0,h}(y) &= x^{h-1}g^h(y)\rho(y),\quad  \Delta m^{0,h} = x^{h-1}\mathbb{E}[Yg^h(Y)].
\end{align*}
\end{lemma}
\begin{proof}
    We proceed by induction. Let $h=1$, then 
    \begin{align*}
    \Delta \rho^{0,1}(y) &= (y-\mu_1)\rho(y) = x^{0}g^1(y)\rho(y)\\
    \Delta m^{0,1} &= \text{Var}(\rho) = \mathbb{E}[Y^2-\mu Y] = x^0\mathbb{E}[Yg^1(Y)]
\end{align*}
Assume the form holds at time step $h$, then for $h+1$,
\begin{align*}
        \Delta \rho^{0,h+1}(y) &= xy\Delta\rho^{0,h}(y) - x\rho(y)\Delta m^{0,h}\\
        &= xy\Big(x^{h-1}g^h\rho(y)\Big) - x\rho(y) x^{h-1}\mathbb{E}[Yg^h(Y)]\\
        &= x^{h}\Big(yg^{h} - \mathbb{E}[Yg^h(Y)]\Big)\rho(y)\\
        &= x^hg^{h+1}(y)\rho(y)\\
     \Delta m^{0,h+1} &= x^h\int_0^1 y g^{h+1}(y)\rho(y) dy \\
     &= x^h \mathbb{E}[Yg^{h+1}(Y)]
\end{align*}
\end{proof}
This form enables us to provide a simple condition for the future rewards to be non-negative. Specifically, for $\Delta m^{0,h} \geq 0$ we need $\mathbb{E}[Yg^h(Y)] \geq 0$. This can be analysed through the properties of the operator $\mathcal{T}$, by splitting the steps of the update. Define the inner product of two functions as
\begin{equation}\label{inner product}
    \langle f,g\rangle := \mathbb{E}[Yf(Y)g(Y)].
\end{equation}
\begin{lemma}\label{lemma: operator is self adjoint}
    The operator $\mathcal{T}$ is self-adjoint with respect to the inner product in Equation \eqref{inner product}.
\end{lemma}
\begin{proof}
\begin{align*}
    \langle \mathcal{T}f,g\rangle  &= \langle Yf- \mathbb{E}[Yf], g\rangle\\
    &= \mathbb{E}[Y(Yf- \mathbb{E}[Yf])g]\\
    &= \mathbb{E}[Y^2fg] - \mathbb{E}[Yf]\mathbb{E}[Yg]\\
    &= \mathbb{E}[Yf(Yg- \mathbb{E}[Yg])]\\
    &= \langle f, Yg- \mathbb{E}[Yg]\rangle\\
    &= \langle f, \mathcal{T}g\rangle.
\end{align*}
\end{proof}

\begin{lemma}\label{lemma: gm,gn is g^m+n}
Let $g^h$ be defined by Equation \eqref{eqn: g iterative definition}. For any $m,n\in \mathbb{N}$,
\begin{align*}
    \mathbb{E}[Yg^mg^n]= \mathbb{E}[Yg^{m+n}].
    \end{align*}
\end{lemma}

\begin{proof}
By Lemma \ref{lemma: operator is self adjoint}, we have
\begin{align*}
    \langle g^{m}, g^n\rangle &= \langle\mathcal{T}^{m-1}g^1, \mathcal{T}^{n-1}g^1 \rangle\\
    &= \langle g^1, \mathcal{T}^{m+n-2}g^1 \rangle\\
    &= \langle g^1, g^{m+n-1} \rangle
\end{align*}
and using the definition of the inner product, for any $k \in \mathbb{N}$
\begin{align*}
    \langle g^{1}, g^k\rangle &= \mathbb{E}[Yg^1g^k]\\
    &= \mathbb{E}[Y(Y-\mu_1)g^k]\\
    &= \mathbb{E}[Y^2g^k]- \mu_1\mathbb{E}[Yg^k]
\end{align*}
From the recursion, we can take expectation over $Yg^{k+1}$, giving the identity
\begin{align*}
    \mathbb{E}[Yg^{k+1}] &= \mathbb{E}[Y^2g^{k}] - \mu_1\mathbb{E}[Yg^{k}]
\end{align*}
Therefore,
\begin{align*}
    \langle g^{1}, g^k\rangle &= \mathbb{E}[Yg^1g^k] = \mathbb{E}[Yg^{k+1}]
\end{align*}
Finally, 
\begin{align*}
    \mathbb{E}[Yg^mg^n]&=\langle g^{m}, g^n\rangle = \langle g^1, g^{m+n-1}\rangle
    = \mathbb{E}[Yg^{m+n}].
\end{align*}
\end{proof}

\begin{lemma}\label{lemma: OFT reward is non-negative for t>1}
    Under the OFT rule, for all $h \geq 1$ the reward difference $\Delta r^{0,h} = b \Delta m^{0,h} \geq 0$.
\end{lemma}
\begin{proof}
    We will separate into odd and even cases. First, let $m=n=l$ in Lemma \ref{lemma: gm,gn is g^m+n}, then
    \begin{align*}
        \Delta m^{0,2l} &= x^{2l-1}\mathbb{E}[Y g^{2l}] = x^{2l-1}\mathbb{E}[Y g^{l}g^l] =  x^{2l-1}\mathbb{E}[Y (g^l)^2] \geq 0
    \end{align*}
    Likewise, let $m=l$ and $n=l+1$. By Lemma \ref{lemma: gm,gn is g^m+n}
        \begin{align*}
        \Delta m^{0,2l+1} &= x^{2l}\mathbb{E}[Y g^{2l+1}]\\
        &= x^{2l}\mathbb{E}[Y g^{l}g^{l+1}]\\
        &= x^{2l}\mathbb{E}[Y g^{l}(Yg^l-\mathbb{E}[Yg^l])]\\
        &=x^{2l}\Big(\mathbb{E}[(Y g^{l})^2] - \mathbb{E}[Y g^{l}]^2\Big)
        \geq 0
    \end{align*}
    Hence, since $b>0$, $\Delta r^{0,h}\geq 0$ for all $h \geq 1$.
\end{proof}
Here, we show how the same result applies to ROFT. 
\begin{lemma}\label{lemma: ROFT reward is non-negative for t>1}
    Under the ROFT rule, for all $h \geq 1$ the reward difference $\Delta r^{0,h} = b \Delta m^{0,h} \geq 0$.
\end{lemma}
\begin{proof}
Under the ROFT partner selection rule,
\begin{align*}
    \Delta\rho^{0,h+1}(y) = (1-x)(1-y)\Delta\rho^{0,h}(y)  + (1-x)\rho(y)\Delta m^{0,h}(x)
\end{align*}
and using the same inductive proof in Lemma \ref{lemma: recursion simplification oft, k=0}, we get
\begin{align*}
    \Delta\rho^{0,h}(y) &= (1-x)^{h-1}f^h(y)\rho(y)\\
    \Delta m^{0,h} &= (1-x)^{h-1}\mathbb{E}[Yf^h(Y)]
\end{align*}
where $f^h(y) = (1-y)f^{h-1}(y) + \mathbb{E}[Yf^{h-1}(Y)]$, with $f^1(y) = y - \mu_1$ and $\mathbb{E}[f^h(Y)] =0$. This means we can write
\begin{align*}
    f^h(y) &= (1-y)f^{h-1}(y) - \mathbb{E}[(1-Y)f^{h-1}(Y)].
\end{align*}
Now consider the transformation $Z = 1-Y$, and the function $\phi^h(z)= -f^h(1-z)$. Then,
\begin{align*}
    \phi^{h}(z) = z\phi^{h-1}(z) - \mathbb{E}[Z\phi^{h-1}(Z)],\quad \phi^1 = z - \mathbb{E}[Z]
\end{align*}
this is the exact same recursion defined for the OFT partner selection rule, and therefore the result of Lemma \ref{lemma: OFT reward is non-negative for t>1} follows. Concretely,
\begin{align*}
    \Delta m^{0,h} &= (1-x)^{h-1}\mathbb{E}[Yf^h(Y)]\\
    &= (1-x)^{h-1}\mathbb{E}[Z\phi^h(Z)] \geq 0.
\end{align*}
Since $b>0$, the result follows.
\end{proof}
We now consider the general case of conditioning on the action at step $k$. Note that by dividing through by $\rho(y)$, the one-step update for OFT becomes
\begin{align*}
    q^{h+1}(y|x) &= xy q^h(y|x) + (1-xm^h(x)),\quad m^h(x) = \mathbb{E}[Yq^h(Y|x)]
\end{align*}
This distribution can then be expressed in terms of the recursive function, $g^h(y)$.
\begin{lemma}
    For all $h$, the distribution $q^h(y|x)$ is given by 
    \begin{align*}
        q^h(y|x) &= 1 + \sum^{h}_{j=1}x^jg^j(y)
    \end{align*}
\end{lemma}
\begin{proof}
    We proceed by induction. Let $h=0$, then 
    \begin{align*}
        q^0(y|x) &= \frac{\rho^0(y|x)}{\rho(y)} = 1.
    \end{align*}
    Assume that the statement holds for $h=k$. Then for $h=k+1$,
    \begin{align*}
        q^{k+1}(y|x)&= xy q^{k}(y|x) + 1-xm^k(x)\\
        &= xy \Big(1+\sum_{j=1}^kx^jg^j(y)\Big)+1-x\mathbb{E}\bigg[Y\Big(1+\sum_{j=1}^kx^jg^j(Y)\Big)\Big]\\
        &= 1+ x(y-\mu_1) + \sum_{j=1}^kx^{j+1}\Big(yg^j(y)-\mathbb{E}[Yg^j(Y)]\Big)\\
        &= 1 + xg^1(y) + \sum_{j=1}^k x^{j+1}g^{j+1}
        (y)\\
        &= 1 + \sum_{j=1}^{k+1} x^jg^j(y)
    \end{align*}
\end{proof}

\subsubsection*{Proof of Proposition \ref{prop: OFT and ROFT have r^h_k geq 0}}
\begin{proof}
Consider the $k$th step of the REINFORCE algorithm. Then the subsequent update under the OFT rule is,
\begin{align*}
    \Delta \rho^{k,k+1} &= y\rho^k(y|x) - m^k(x) \rho(y)\\
    g^{k,k+1} &= y q^k(y|x) - m^k(x)\\
    &= y \bigg(1 + \sum^{k}_{j=1}x^jg^j(y) \bigg) - \bigg(\mu_1  + \sum^{k}_{j=1}x^j\mathbb{E}[Yg^j(Y)]\bigg)\\
    &= y- \mu_1 + \sum^{k}_{j=1}x^j\Big(yg^j(y)-\mathbb{E}[Yg^j(Y)]\Big)\\
    &= y-\mu_1 + \sum^k_{j=1} x^j g^{j+1}(y)\\
    &=\sum_{j=0}^k x^jg^{j+1}(y)
\end{align*}
Then for $h>k+1$, we have
\begin{align*}
    \Delta \rho^{k,h+1}(y) &= xy \Delta \rho^{k,h}(y) -x\rho(y)\Delta m^{k,h}  
\end{align*}
and therefore the same operator $\mathcal{T}$ applies. Consequently, 
\begin{align*}
    g^{k,h}(y) &= \mathcal{T}^{h-k-1} g^{k,k+1}(y)\\
    &= \sum_{j=0}^k x^j \mathcal{T}^{h-k-1}g^{j+1}(y)\\
    &= \sum_{j=0}^k x^j g^{h-k+j}(y)
\end{align*}
The difference in the conditional means at this step is then
\begin{align*}
    \Delta m^{k,h} &= x^{h-k-1}\mathbb{E}[Yg^{k,h}(Y)]\\
    &= x^{h-k-1} \mathbb{E}[Y\sum_{j=0}^k x^j g^{h-k+j}(Y)]\\
    &= \sum_{j=0}^k x^{h-k+j-1} \mathbb{E}[Y g^{h-k+j}(Y)]\\
    &= \sum_{j=0}^k \Delta m^{0,h-k+j} \geq 0
\end{align*}
where the final line follows by Lemma \ref{lemma: OFT reward is non-negative for t>1}. Note that the proof has an identical structure for the ROFT rule, invoking Lemma \ref{lemma: ROFT reward is non-negative for t>1} in the final step.
\end{proof}

\subsubsection*{Proof of Corollary \ref{cor: H=2 is sufficient}}
\begin{proof}
    For each step $k$ of the REINFORCE algorithm, we obtain a reward difference of $\Delta r^{k,k}=-c$. The future rewards must compensate for this. Consider the actions conditioned at step $k$. By the proof of Proposition \ref{prop: OFT and ROFT have r^h_k geq 0}, at step $k+1$ 
    \begin{align*}
        \Delta r^{k,k+1}(x) &= b \Delta m^{k,k+1}(x)\\
        &= b \sum_{j=0}^k \Delta m^{0,j+1} \\
        &\geq b\Delta m^{0,1}.
    \end{align*}
    Since we have $\Delta m^{0,1} = \text{Var}(\rho)$, and for all $h>k+1$ the reward difference $\Delta r^{k,h}(x) \geq 0$, then 
    \begin{align*}
        \sum^{H-1}_{h=0} \Delta G^h[\rho] &= \sum^{H-1}_{k=0} \sum_{h=k}^{H-1} \Delta r^{k,h}\\ &\geq \sum^{H-1}_{k=0}  \Big(\Delta r^{k,k} + \Delta r^{k,k+1}\Big)\\ &\geq  \sum^{H-1}_{k=0}  \Big(-c + b\text{Var}(\rho)\Big) \\&\geq (H-1)(b\text{Var}(\rho) - c) -c.
    \end{align*}
    which is increasing in $H$ when $\Delta G[\rho] >0$ for $H=2$.
\end{proof}

\subsection{Mean Dynamics}

\subsubsection*{Proof of Theorem \ref{thm: convergence to pure defection}}
\begin{proof}
    The characteristic flow is generated by 
    \begin{align*}
        \frac{d}{dt}X_t(x) = -4\alpha cX^2_t(x)(1-X_t(x))^2, \quad X_0(x) = x_0,
    \end{align*}
    This is a separable ODE, which integrating gives
    \begin{align*}
         \int \frac{1}{X^2(1-X)^2} dX &= \int -4\alpha c\; dt\\
          \frac{1}{1-X_t}-\frac{1}{X_t}+ 2 \ln |\frac{X_t}{1-X_t}| &= -4\alpha ct + C
    \end{align*}
    Let 
    \begin{align*}
            F(x) &= \frac{1}{1-x}-\frac{1}{x}+ 2 \ln |\frac{x}{1-x}|
    \end{align*}
    then the solution at time $t$ is given by
    \begin{align*}
        F(X_t) = F(x_0) -4\alpha ct.
    \end{align*}
    For $c>0$ and $\alpha>0$, then as $t\rightarrow\infty$ we have $F(X_t) \rightarrow -\infty$ for every $x_0 \in (0,1)$. Since $F(x) \rightarrow-\infty$ as $x \rightarrow0$, we have that $X_t \rightarrow 0$.  Since the velocity field $u(x)=-4\alpha cx^2(1-x)^2\in C^1([0,1])$, the solution to the continuity equation is given by the pushforward of the initial point under the characteristic flow:
    \begin{align*}
        \rho(t,\cdot)=(X_t)_\#\rho_0.
    \end{align*}
    Then for any bounded test function $\varphi\in C([0,1])$,
    \begin{align*}
        \int_0^1 \varphi(y)\rho(t,y)\; dy = \int_0^1\varphi(X(t,x))\rho_0(x)dx
    \end{align*}
    Since $X(t,x) \rightarrow 0$ pointwise for all $x\in (0,1)$, then $\varphi(X(t,x))\rightarrow\varphi(0)$. Moreover,
    \begin{align*}
        |\varphi(X(t,x))\rho_0(x)|&\leq \|\varphi\|_\infty|\rho_0(x)|
    \end{align*}
    By the dominated convergence theorem,
    \begin{align*}
        \int_0^1\varphi(X(t,x))\rho_0(x)dx \rightarrow\varphi(0) \int_0^1\rho_0(x) dx = \varphi(0)
    \end{align*}
    which means
    \begin{align*}
        \int_0^1 \varphi(y)\rho(t,y)\; dy \rightarrow \varphi(0)
    \end{align*}
    and therefore $\rho(t) \rightharpoonup \delta_0$.  
\end{proof}

\begin{lemma}\label{lem: advection spatially lipschitz}
    The non-local advection is spatially Lipschitz and sub-linear. That is, there exists an $L_1, M \in \mathbb{R}$ such that 
    \begin{align*}
        |u[\rho]( x) - u[\rho](y)| \leq L_1| x-y|,\qquad |u[\rho]( x)|\leq M(1+| x|).
    \end{align*}
\end{lemma}
\begin{proof}
    The spatial derivative is 
    \begin{align*}
        \partial_ x u[\rho]( x) = 4\alpha \Delta G[\rho]\cdot x(1-x)(1-2x)
    \end{align*}
    which is uniformly bounded by $\Delta G[\rho]$ on $[0,1]$. Therefore by the mean value theorem,
    \begin{align*}
        |u[\rho]( x) - u[\rho](y)| &\leq \text{sup}_{z}  |\partial_x u[\rho](z)| \;| x-y|\\
        &\leq 2\alpha|\Delta G[\rho]|| x-y|
    \end{align*}
    Since the episode length and per-game reward are bounded, $|\Delta G[\rho]|$ is uniformly bounded. Therefore exists an $L_1$ such that $u[\rho](y)$ is spatially Lipschitz with constant $L_1$. The condition for sub-linearity follows from 
    \begin{align*}
        |u[\rho]( x)| &= | 2\alpha\Delta G[\rho] x^2(1- x)^2|\\
        &\leq2\alpha|\Delta G[\rho]|| x^2(1- x)^2|\\
        &\leq \frac{\alpha}{8}|\Delta G[\rho]| 
    \end{align*}
    which is uniformly bounded.
\end{proof}

\begin{lemma}\label{lem: wasserstein}
    For every $\rho_1,\rho_2 \in \mathcal{P}$, there exists a positive constant $L_2 \in \mathbb{R}$ such that 
    \begin{align*}
        \|u[\rho_1](\cdot) - u[\rho_2](\cdot)\|_{\mathcal{C}([0,1])} \leq L_2 W_1(\rho_1,\rho_2).
    \end{align*}
\end{lemma}
\begin{proof}
    \begin{align*}
        |u[\rho_1](x) - u[\rho_2](x)| &= |2\alpha\Delta G[\rho_1]x^2(1-x)^2 -2\alpha\Delta G[\rho_2]x^2(1-x)^2|\\
        &= 2\alpha|x^2(1-x)^2 (\Delta G[\rho_1]-\Delta G[\rho_2])|\\
        &\leq \frac{\alpha}{8}|\Delta G[\rho_1]-\Delta G[\rho_2]|
    \end{align*}
    Therefore, 
    \begin{align*}
    \|u[\rho_1](\cdot) - u[\rho_2](\cdot)\|_{C^0([0,1])} \leq \frac{\alpha}{8}|\Delta G[\rho_1]-\Delta G[\rho_2]|
    \end{align*}
    It suffices to show that 
    \begin{align*}
    |\Delta G[\rho_1]-\Delta G[\rho_2]|  &\leq C W_1(\rho_1,\rho_2)
    \end{align*}
    for some constant $C \in \mathbb{R}$. Using the Kantorovich-Rubinstein duality \cite{villani}, the Wasserstein distance with $p=1$ can be expressed as 
    \begin{align*}
        W_1(\rho_1,\rho_2) = \sup \bigg\{\int_{[0,1]} f(x)d(\rho_1-\rho_2)\; \big|\; f \in \mathcal{C}^1,\; f: [0,1] \rightarrow\mathbb{R},\; \text{Lip}(f) \leq 1\bigg\}.
    \end{align*}
    For any integer $i$, the $i$th moment is 
    \begin{align*}
        m_i &= \int_0^1 x^i d\rho(x)
    \end{align*}
    Let $f(x) = \frac{x^i}{i}$, then $f$ is continuous and has Lipschitz constant 1. By the Wasserstein metric, 
    \begin{align*}
        \frac{1}{i}|m_i(\rho_1) - m_i(\rho_2)| \leq W_1(\rho_1,\rho_2).
    \end{align*}
    
    Therefore,
    \begin{align*}
       |\Delta G[\rho_1] - \Delta G[\rho_2]| &= b|\text{Var}(\rho_1) - \text{Var}(\rho_2)|\\
       &\leq b|m_2(\rho_1) - m_2(\rho_2)| + b |m_1(\rho_1)^2 - m_1(\rho_2)^2|\\
       &\leq 2b W_1(\rho_1,\rho_2) + b|m_1(\rho_1)+m_1(\rho_2)||m_1(\rho_1)-m_1(\rho_2)|\\
       &\leq 2bW_1(\rho_1,\rho_2) + b \cdot 2 \cdot W_1(\rho_1,\rho_2)\\
       &= 4bW_1(\rho_1,\rho_2)
    \end{align*}
    Therefore the result holds with $L_2 = \frac{b}{4}$.
\end{proof}

\subsubsection{Proof of Proposition \ref{prop: existence to continuity equation}}
\begin{proof}
    By Theorem 2 in \cite{bonnet_pontryagin_2019}, it suffices to show the conditions of Lemma \ref{lem: advection spatially lipschitz} and Lemma \ref{lem: wasserstein} are satisfied. These require that $\Delta G[\rho]$ is bounded, which is always satisfied for $\rho_0 \in \mathcal{P}([0,1)]$.
\end{proof}

The characteristic ODE admits a solution $X_t$
\begin{align*}
    \frac{d}{dt}X_t&=2\alpha\Delta G[\rho]X_t^2(1-X_t)^2\\
    \int\frac{1}{X^2(1-X)^2}dX&=2\alpha\int \Delta G[\rho]dt\\
    F(X_t) &= F(x_0) + 2\alpha K(t) 
\end{align*}
where $F(x)$ is defined in the proof of Theorem \ref{thm: convergence to pure defection}. Clearly, the time evolution is uniquely determined by $K$. Since $F'(x)>0$ for all $x \in (0,1)$, $F$ is strictly increasing and therefore invertible on the domain. Therefore the solution can be expressed by
\begin{align*}
    X_K(x) = F^{-1}(F(x_0) + 2\alpha K).
\end{align*}
The next result proves that under this flow, the cooperation levels increase when the initial population has sufficient variance.
\subsubsection*{Proof of Theorem \ref{thm: convergence to cooperation}}
\begin{proof}
    Define the autonomous equations
    \begin{align*}
        K'&= h(K), \\
        h(K) &= b \text{Var}\big((X_K)_\#\rho_0\big) - 2c
    \end{align*}
    The initial condition means
    \begin{align*}
        h(0) &= b\text{Var}(\rho_0)-2c > 0
    \end{align*}
    Note that for every $x>0$
    \begin{align*}
        \lim_{K\rightarrow\infty} X_K(x) \rightarrow1, \quad \partial_K X_K = 2\alpha X_K^2(1-X_K)^2 \leq \alpha/8,
    \end{align*}
    therefore $X_K$ is uniformly Lipschitz in $K$ and bounded on $[0,1]$. In particular,
    \begin{align*}
        |X_{K_1}(x) -X_{K_2}(x)| \leq \frac{\alpha}{8}|K_1-K_2|,\quad \forall x \in [0,1].  
    \end{align*}
    For $X \sim \rho_0$, the first moment is
    \begin{align*}
        |m_1(K_1) -m_1(K_2)| &= |\mathbb{E}[X_{K_1}(X)] -\mathbb{E}[X_{K_2}(X)]|\\
        &\leq \mathbb{E}[|X_{K_1}(X) -X_{K_2}(X)]|\\
        &\leq \frac{\alpha}{8}|K_1-K_2|.
    \end{align*}
    It follows that the second moment is Lipschitz:
    \begin{align*}
        |m_2(K_1) -m_2(K_2)| &= |\mathbb{E}[X_{K_1}^2(X)] -\mathbb{E}[X^2_{K_2}(X)]|\\
        &= |\mathbb{E}[X_{K_1}^2(X) -X^2_{K_2}(X)]|\\
        &\leq \mathbb{E}[|X_{K_1}^2(X) -X^2_{K_2}(X)|]\\
        &\leq \mathbb{E}[2|X_{K_1}(X) -X_{K_2}(X)|]\\
        &\leq \frac{\alpha}{4}|K_1-K_2|
    \end{align*}
    
    Therefore the mean and variance are Lipschitz continuous in $K$, which means there is a unique solution to the IVP $K'=h(K), K(0)=0$. Under the limit as $K\rightarrow\infty$, $\mathbb{E}[X^2_K(X)]\rightarrow1$, so $\text{Var}(K) = \mathbb{E}[X^2_K(X)]-\mathbb{E}[X_K(X)]^2\rightarrow 0$. The left and right limits of $h$ are
    \begin{align*}
        h(0) >0,\qquad \lim_{K\rightarrow\infty} h(K) = -2c <0,
    \end{align*}
    so by the intermediate value theorem ($h$ is Lipschitz continuous), there exists a $K^*=\inf\{K>0\; | h(K)=0 \}$. For all $K < K^*,\; h(K) >0$ therefore $K$ is increasing. Therefore, $K \rightarrow K^*$. 
\end{proof}

\subsubsection*{Proof of Proposition \ref{prop: increase cooperation for some finite time}}
\begin{proof}
Let $f \in C^2([0,1])$ be an arbitrary test function, then
\begin{align*}
    \frac{d}{dt}\mathbb{E}[f(X_t)] = \mathbb{E}[\mathcal{L}f(X_t)]
\end{align*}
where $\mathcal{L}$ is the infinitesimal generator given by
\begin{align*}
    \mathcal{L}f(x) = A_\rho(t,x)f'(x) + \frac{1}{2}B_\rho^2(t,x) f''(x).
\end{align*}
In particular, letting $f(x)=x$ gives the evolution of the mean as
\begin{align*}
    \frac{d}{dt}\mathbb{E}[X_t] = \mathbb{E}[A_\rho(t,x)]
\end{align*}
which gives the ODE
\begin{align*}
    \frac{d}{dt}m_1(t) = \int_0^1 A_\rho(t,x)\rho(t,x) dx  &= 2\alpha \int_0^1 x^2(1-x)^2\Delta G[\rho] \rho(t,x) dx + O(\alpha^2).
\end{align*}
    To show that this is increasing, we begin by showing the $O(\alpha^2)$ term is uniformly bounded by a constant $C$. Each variance term is bounded by 
    \begin{align*}
        \text{Var}(U^h) & \leq \mathbb{E}[(U^h)^2] \leq (H(b+c)+|\beta|)^2
    \end{align*}
    and the covariance is bounded by 
    \begin{align*}
        |\text{Cov}(U^h,h^k)|  &\leq \sqrt{\text{Var}(U^h)\text{Var}(U^k)} \leq (H(b+c)+|\beta|)^2.
    \end{align*}
    Therefore, we have
    \begin{align*}
        |\Sigma_{CC}| &= \Big|\alpha^2\sum_{h=0}^{H-1} \text{Var}(U^h) + 2\alpha^2 \sum_{h<k}^{H-1}\text{Cov}(U^h,U^k)\Big|\\
        &\leq \alpha^2(H + H(H-1))(H(b+c)+|\beta|)^2\\
        &= \alpha^2H^2(H(b+c)+|\beta|)^2
    \end{align*}
    This give a simple bound on the $O(\alpha^2)$ term as
    \begin{align*}
    |2x(1-x)(1-2x)\Sigma_{CC}|  &\leq \frac{1}{4}\alpha^2H^2(H(b+c)+|\beta|)^2.
\end{align*}
Define
\begin{align*}
    I(t) &= \int_0^1 x^2(1-x)^2\Delta G[\rho] \rho(t,x) dx
\end{align*}
and let $\alpha^*=\frac{4I(0)}{H^2(H(b+c)+|\beta|)^2}$. Since $I(t)$ is continuous and $I(0) >0$, then there exists a $T$ such that
\begin{align*}
    I(t) > I(0)/2 >0 \qquad \forall t \in [0,T].
\end{align*}
Hence on the interval $[0,T]$, 
\begin{align*}
    \frac{d}{dt}m_1(t) &= 2 \alpha I(t) +O(\alpha^2)\\
    &\geq\alpha I(0) -\frac{1}{4}\alpha^2H^2(H(b+c)+|\beta|)^2\\
    &>0
\end{align*}
for all $\alpha < \alpha^*$.
\end{proof}

To show the regularity conditions required for a well-posed solution to the steady-state equation, we provide the following Lipschitz continuity result on moments
\begin{align*}
    |m_i(\eta) - m_i(\nu)| &= |\int_0^1 y^id (\eta -\nu) (y)| \leq i W_1(\eta,\nu) .
\end{align*}
Next, we use this moment regularity to show the conditional reward terms are also Lipschitz.
\begin{lemma}
    The drift $A^\varepsilon_\eta$ and diffusion $B^{2,\varepsilon}_\eta$ are Lipschitz continuous in both $x$ and $\eta$. Specifically,
    \begin{align*}
        \|A^\varepsilon_\eta- A^\varepsilon_\nu\|_\infty \leq L_A W_1(\eta,\nu), \quad  \|B^{2,\varepsilon}_\eta- B^{2,\varepsilon}_\nu\|_\infty \leq L_B W_1(\eta,\nu),
    \end{align*}
    and
    \begin{align*}
        |A^\varepsilon_\eta(x)- A^\varepsilon_\eta(y)| \leq L_A |x-y|, \quad  |B^{2,\varepsilon}_\eta(x)- B^{2,\varepsilon}_\eta(y)| \leq L_B |x-y|,
    \end{align*}
    where $L_A,L_B \in \mathbb{R}^+$ and are independent of $\varepsilon$.
\end{lemma}
\begin{proof}
    First note that $\Delta G[\rho]$ is a polynomial in $m_1$ and $m_2$, and $\Sigma$ is a polynomial in $x, m_1, m_2, m_3$. Since the product of polynomials is a polynomial, we have $A_\eta^\varepsilon = f(x,m_1,m_2,m_3)$ is a polynomial such that $x,m_1,m_2,m_3 \in[0,1]$. Therefore, $f$ is Lipschitz in each of its arguments such that 
    \begin{align*}
        L_f = \sup_{[0,1]^4}\| \nabla f(x,m_1,m_2,m_3)\|_\infty < \infty.
    \end{align*}
    Hence,
    \begin{align*}
        |f(x,m_1(\eta),m_2(\eta),m_3(\eta)) - f(x,m_1(\nu),m_2(\nu),m_3(\nu))| &\leq L_f (|m_1(\eta) - m_1(\nu)| + |m_2(\eta) - m_2(\nu)|\\ &+ |m_3(\eta) - m_3(\nu)|)\\
        &\leq L_f (W_1(\eta,\nu) + 2W_1(\eta,\nu) + 3W_1(\eta,\nu))\\
        &= 6L_fW_1(\eta,\nu)
    \end{align*}
    For Lipschitz continuity in $x$, note that since $f$ is a polynomial, its derivative is bounded. In particular,
    \begin{align*}
        \sup_{[0,1]^4}|\partial_x f| < \infty
    \end{align*}
    and therefore for any $\eta$,
    \begin{align*}
        |A_\eta(x)-A_\eta(y)| \leq \sup_{[0,1]^4}|\partial_x f| |x-y|
    \end{align*}
    so $A_\eta$ is Lipschitz in $x$. Taking $L_A$ as the maximum of each Lipschitz bound gives the result. The proof for $B^{2,\varepsilon}_\eta$ is identical, since this is also a polynomial.
\end{proof}
Note that since the drift and diffusion are Lipschitz and the domain is bounded, then both the drift and diffusion are bounded above such that $A^\varepsilon_\eta(x) \leq M_A$ and $B^{2,\varepsilon}_\eta(x) \leq M_B$. The next result shows the iterative operator is well-defined.
\begin{lemma}\label{lem: F operator is well defined}
    The operator $\mathcal{F}^\varepsilon[\eta]: S^\varepsilon \rightarrow S^\varepsilon$ is well-defined.
\end{lemma}
\begin{proof}
    Note that 
    \begin{align*}
        \bigg|\frac{2A^\varepsilon_\eta(x)}{B^{2,\varepsilon}_\eta(x)}\bigg| &\leq \bigg|\frac{2A^\varepsilon_\eta(x)}{\varepsilon}\bigg| \leq \frac{2M_A}{\varepsilon}
    \end{align*}
    so for every $x \in [0,1]$
    \begin{align*}
        \bigg |\int_0^x\frac{2A^\varepsilon_\eta(y)}{B^{2,\varepsilon}_\eta(y)} dy \bigg |&\leq \frac{2M_A}{\varepsilon}
    \end{align*}
    and 
    \begin{align*}
        e^{-\frac{2M_A}{\varepsilon}} \leq \exp\bigg(\int_0^x\frac{2A^\varepsilon_\eta(y)}{B^{2,\varepsilon}_\eta(y)} dy\bigg) \leq e^{\frac{2M_A}{\varepsilon}}
    \end{align*}
    This means 
    \begin{align*}
        0 < w^\varepsilon_\eta(x) \leq \frac{e^{\frac{2M_A}{\varepsilon}}}{\varepsilon}
    \end{align*}
    We have that $w^\varepsilon_\eta(x) \in C[0,1]$, $w^\varepsilon_\eta(x)>0$ and the integral across the domain is finite and strictly positive. Moreover,
    \begin{align*}
        \int_0^1 \mathcal{F}^\varepsilon[\eta](x) dx = 1.
    \end{align*}
    Therefore $\mathcal{F}^\varepsilon:S^\varepsilon \rightarrow S^\varepsilon$ is a well-defined mapping. Since $B^{2,\varepsilon}_\eta(x) \leq M_B + \varepsilon$, we also have the bound
    \begin{align*}
        \frac{1}{M_B + \varepsilon} e^{-\frac{2M_A}{\varepsilon}} &\leq w^\varepsilon_\eta(x) \leq \frac{e^{\frac{2M_A}{\varepsilon}}}{\varepsilon}\\
        \frac{1}{M_B + \varepsilon} e^{-\frac{2M_A}{\varepsilon}} &\leq \int_0^1w^\varepsilon_\eta(x) dx \leq \frac{e^{\frac{2M_A}{\varepsilon}}}{\varepsilon}
    \end{align*}
    hence $0 \leq \mathcal{F}^\varepsilon[\eta](x) \leq M^\varepsilon$ for some constant $M^\varepsilon$ which is independent of $\eta$.
\end{proof}

\begin{lemma}\label{lemma: F is relatively compact}
    The set $\mathcal{F}^\varepsilon(S):= \{\mathcal{F}^\varepsilon[\eta]:\eta \in S^\varepsilon\} \subset C([0,1])$ is relatively compact.
\end{lemma}
\begin{proof}
    From the Arzela-Ascoli Theorem \cite{rudin_principles}, it suffices to show that the set $\mathcal{F}^\varepsilon(S^\varepsilon)$ is equibounded and equicontinuous. Firstly, by Lemma \ref{lem: F operator is well defined}, the set is uniformly bounded by $M^\varepsilon$. 
    Now let
    \begin{align*}
        \psi^\varepsilon_\eta(x)&= \int_0^x \frac{2A^\varepsilon_\eta(y)}{B^{2,\varepsilon}_\eta(y)}dy
    \end{align*}
    Then, 
    \begin{align*}
        |\partial_x\psi^\varepsilon_\eta(x)| &= \bigg|\frac{2A^\varepsilon_\eta(x)}{B^{2,\varepsilon}_\eta(x)}\bigg|\\
        &\leq \frac{2M_A}{\varepsilon}
    \end{align*}
    To pass the bound to the operator, we use that $B^{2,\varepsilon}_\eta(x)$ is uniformly Lipschitz in $x$ and has a lower bound $\varepsilon$ to get
    \begin{align*}
        |\partial_x \log(B^{2,\varepsilon}_\eta)| &= |\frac{\partial_xB^{2,\varepsilon}_\eta(x)}{B^{2,\varepsilon}_\eta(x)}|\leq \frac{L_B}{\varepsilon}.
    \end{align*}
    then since $w^\varepsilon_\eta(x) = e^{\psi^\varepsilon_\eta(x)}/B^{2,\varepsilon}_\eta(x)$,
    \begin{align*}
        \log w^\varepsilon_\eta(x) &= \psi^\varepsilon_\eta(x) - \log(B^{2,\varepsilon}_\eta(x))
    \end{align*}
    has a uniform Lipschitz bound. By Lemma \ref{lem: F operator is well defined}, the normalisation is bounded away from zero. Therefore, $\mathcal{F}^\varepsilon[\eta](x)$ has a uniform Lipschitz bound. Hence, $\mathcal{F}^\varepsilon[\eta](x)$ is uniformly bounded and equicontinuous and therefore relatively compact in $C([0,1])$.
\end{proof}

\begin{proposition}\label{prop: F is Lipschitz continuous}
    The mapping $\mathcal{F}^\varepsilon: S^\varepsilon \rightarrow S^\varepsilon$ is Lipschitz continuous.
\end{proposition}
\begin{proof}

 Consider the integrand
  \begin{align*}
     \bigg|\frac{2A^\varepsilon_\eta}{B^2_\eta+\varepsilon} - \frac{2A^\varepsilon_\nu}{B^2_\nu+\varepsilon}\bigg| &\leq \bigg|\frac{2(A^\varepsilon_\eta-A^\varepsilon_\nu)}{B^2_\eta+\varepsilon}\bigg| + \frac{2|A^\varepsilon_\nu||B^2_\eta-B^2_\nu|}{|(B^2_\eta+\varepsilon)(B^2_\nu+\varepsilon)|}\\
     &\leq \frac{2|A^\varepsilon_\eta-A^\varepsilon_\nu|}{\varepsilon} + \frac{2M_A|B^2_\eta-B^2_\nu|}{\varepsilon^2}\\
     &\leq (\frac{2L_A}{\varepsilon} + \frac{2M_AL_B}{\varepsilon^2})W_1(\eta,\nu)
    \end{align*}   
    where in the final line we take the supremum over $x$. This implies
    \begin{align*}
        \|\psi^\varepsilon_\eta -\psi^\varepsilon_\nu\|_\infty &\leq (\frac{2L_A}{\varepsilon} + \frac{2M_AL_B}{\varepsilon^2})W_1(\eta,\nu)
    \end{align*}
    Now take
    \begin{align*}
        |w^\varepsilon_\eta(x)-w^\varepsilon_\nu(x)\big| &= \bigg|\frac{\exp{(\psi^\varepsilon_\eta)}}{B^{2,\varepsilon}_\eta}-\frac{\exp{(\psi^\varepsilon_\nu)}}{B^{2,\varepsilon}_\nu}\bigg|\\
        &\leq \bigg|\frac{\exp{(\psi^\varepsilon_\eta)}-\exp{(\psi^\varepsilon_\nu)}}{B^{2,\varepsilon}_\eta}\bigg|+\frac{|\exp{(\psi^\varepsilon_\nu)||B^{2,\varepsilon}_\eta-B^{2,\varepsilon}_\nu|}}{|(B^{2,\varepsilon}_\eta)(B^{2,\varepsilon}_\nu)|}\\
        &\leq \frac{\exp(\max\{\psi^\varepsilon_\eta, \psi^\varepsilon_\nu\})|\psi^\varepsilon_\eta-\psi^\varepsilon_\nu|}\varepsilon + \frac{\exp(\psi^\varepsilon_\eta)L_B W_1(\eta,\nu)}{\varepsilon^2}\\
        &\leq e^{\frac{2M_A}{\varepsilon}}\Big(\frac{2L_A}{\varepsilon^2} + \frac{2M_AL_B}{\varepsilon^3} + \frac{L_B}{\varepsilon^2}\Big) W_1(\eta,\nu)
    \end{align*}
    where in the penultimate line we use mean-value theorem, and in the final line the uniform upper bound on $\psi^\varepsilon_\eta$ and Lipschitz bound above. Denote this value by $C_w$, then normalisation has the same Lipschitz bound. Hence
    \begin{align*}
        \big| \mathcal{F}^\varepsilon[\eta]-\mathcal{F}^\varepsilon[\nu] \big| &= \bigg|\frac{w^\varepsilon_\eta(x)}{\int_0^1 w^\varepsilon_\eta(x)dx} - \frac{w^\varepsilon_\nu(x)}{\int_0^1 w^\varepsilon_\nu(x)dx}\bigg| \\
        &\leq \frac{\big|w^\varepsilon_\eta(x)-w^\varepsilon_\nu(x)\big|}{\big|\int_0^1 w^\varepsilon_\eta(x)dx\big|} + \frac{|w^\varepsilon_\nu(x)||\int_0^1 w^\varepsilon_\eta(x)-w^\varepsilon_\nu(x)dx|}{|\big(\int_0^1 w^\varepsilon_\eta(x)dx\big)(\int_0^1 w^\varepsilon_\nu(x)dx)|}\\
        \big\| \mathcal{F}^\varepsilon[\eta]-\mathcal{F}^\varepsilon[\nu] \big\|_\infty&\leq (M_B+ \varepsilon)e^{\frac{2M_A}{\varepsilon}}C_w W_1(\eta,\nu) + \frac{(M_B+ \varepsilon)^2}{\varepsilon}e^{\frac{6M_A}{\varepsilon}}C_w W_1(\eta,\nu)\\
        &:= L_\mathcal{F}W_1(\eta,\nu) 
    \end{align*}
\end{proof}

\subsubsection*{Proof of Proposition \ref{prop: existence to regularised problem}}
\begin{proof}
    First note $S^\varepsilon$ is a non-empty, closed, bounded and convex subset of a Banach space. A solution to the steady state equation is uniquely expressed by the fixed point of the mapping $\mathcal{F}^\varepsilon:S^\varepsilon\rightarrow S^\varepsilon$. By Lemma \ref{lem: F operator is well defined}, this operator is well-defined and an invariant mapping. By Proposition \ref{prop: F is Lipschitz continuous} it is continuous and by Lemma \ref{lemma: F is relatively compact} its relatively compact in $C([0,1])$. Applying Schauder's Fixed Point Theorem \cite{shapiro_fixed_point}, there exists at least one fixed point of the mapping $\mathcal{F}^\varepsilon:S^\varepsilon\rightarrow S^\varepsilon$, and therefore a solution to the regularised steady state equation.
\end{proof}

\subsubsection*{Proof of Theorem \ref{thm: existence to un-regularised problem}}
\begin{proof}
    Let $\mu^\varepsilon(dx) = \rho^\varepsilon(x)dx$. Since $[0,1]$ is compact, the space of probability measures $\mathcal{P}([0,1])$ is sequentially compact under weak convergence.  Along any sequence $\varepsilon_n \rightarrow0$, there exists a subsequence and measure $\mu \in \mathcal{P}([0,1])$ such that \cite{billing_probability_measures}
    \begin{align*}
        \mu^{\varepsilon_{n_k}} \rightharpoonup \mu. 
    \end{align*}
    The weak form of the regularised steady-state equation at the stationary distribution $\mu^\varepsilon$ is
    \begin{align*}
        \int_0^1 \Big[A_{\mu^\varepsilon}(x) \varphi'(x) + \frac{1}{2}(B^2_{\mu^\varepsilon}(x)+\varepsilon)\varphi''(x)\Big] d\mu^{\varepsilon}(x) &= 0
    \end{align*}
    for all $\varphi \in C^2([0,1])$ such that $\varphi'(0) = \varphi'(1)=0$. Along the subsequence, we have
    \begin{align*}
        \int_0^1 \Big[A_{\mu^{\varepsilon_{n_k}}}(x) \varphi'(x) + \frac{1}{2}(B^2_{\mu^{\varepsilon_{n_k}}}(x)+\varepsilon_{n_k})\varphi''(x)\Big] d\mu^{\varepsilon_{n_k}}(x) &= 0.
    \end{align*}
    Since the drift, $A_\mu(x)$, and diffusion terms $B^2_\mu(x)$ are continuous and only depend on finitely many bounded moments, weak convergence implies 
    \begin{align*}
        A_{\mu^{\varepsilon_{n_k}}}(x) \varphi'(x) + \frac{1}{2}(B^2_{\mu^{\varepsilon_{n_k}}}(x)+\varepsilon_{n_k})\varphi''(x) \to A_{\mu}(x) \varphi'(x) + \frac{1}{2}B^2_{\mu}(x)\varphi''(x)
    \end{align*}
    uniformly on $[0,1]$. We can take the limit as $k\rightarrow \infty$ and obtain
        \begin{align*}
        \int_0^1 \Big[A_{\mu}(x) \varphi'(x) + \frac{1}{2}B^2_{\mu}(x)\varphi''(x)\Big]d\mu(x) &= 0
    \end{align*}
    which is the weak form of the stationary unregularised problem.
\end{proof}

\section{Empirical Study}\label{appendix: empirical study}
All simulations and numerical solutions have been run on a laptop with 12th Gen Intel(R) Core(TM) i7-1260P CPU. The total compute time for the experiments was approximately 40 hours. For the learning rate time scaling, the simulation and numerical solution for $\alpha = 0.001$ was run for ten times longer than for the $\alpha = 0.01$ case.

\begin{figure}[h!]
     \centering
    \includegraphics[width=0.9\textwidth]{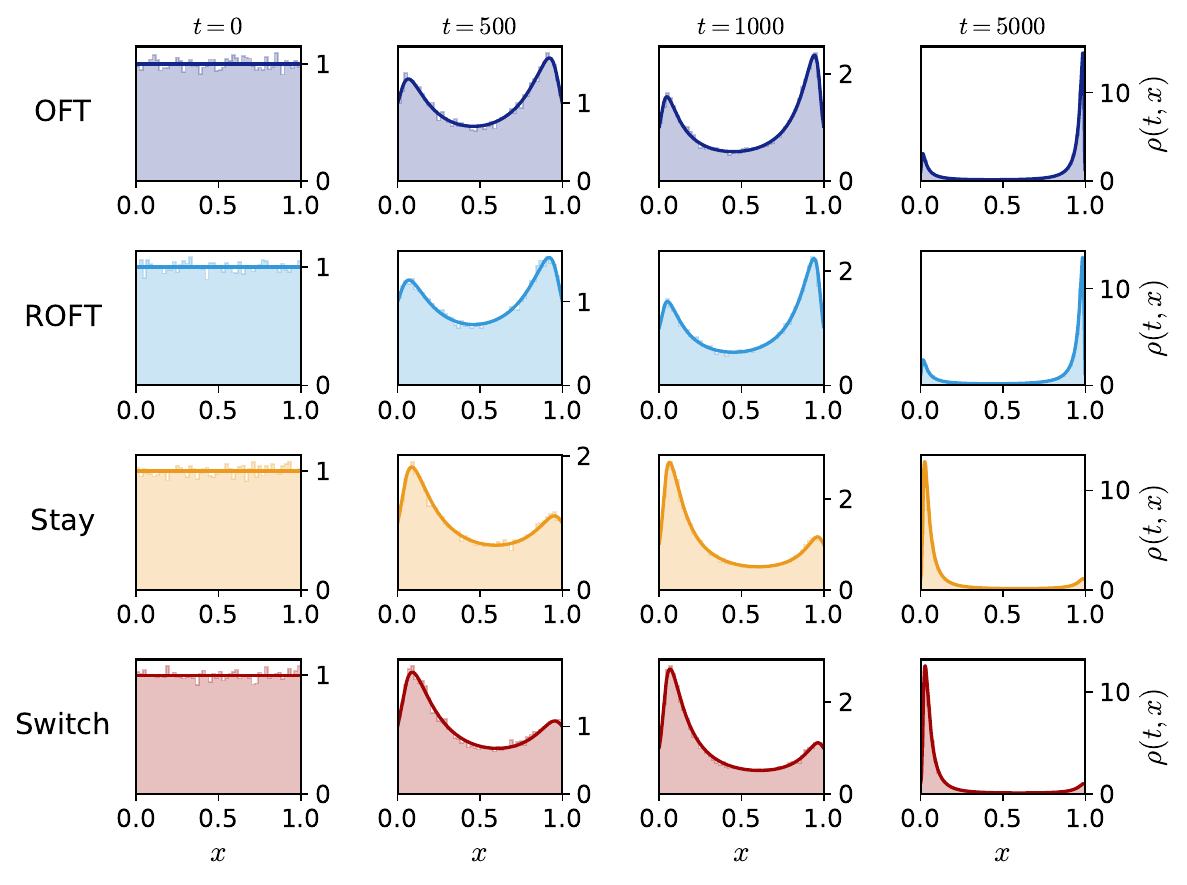}
     \caption{Evolution of the strategy distribution where the population is initialised with $X\sim \text{Uni}(0,1)$ for the 4 rules. The theoretical solution (solid line) matches the simulations (histogram).}
     \label{fig: comparison of evolutions for 4 PS rules, uniform}
\end{figure}

\begin{figure}[h!]
     \centering
    \includegraphics[width=0.9\textwidth]{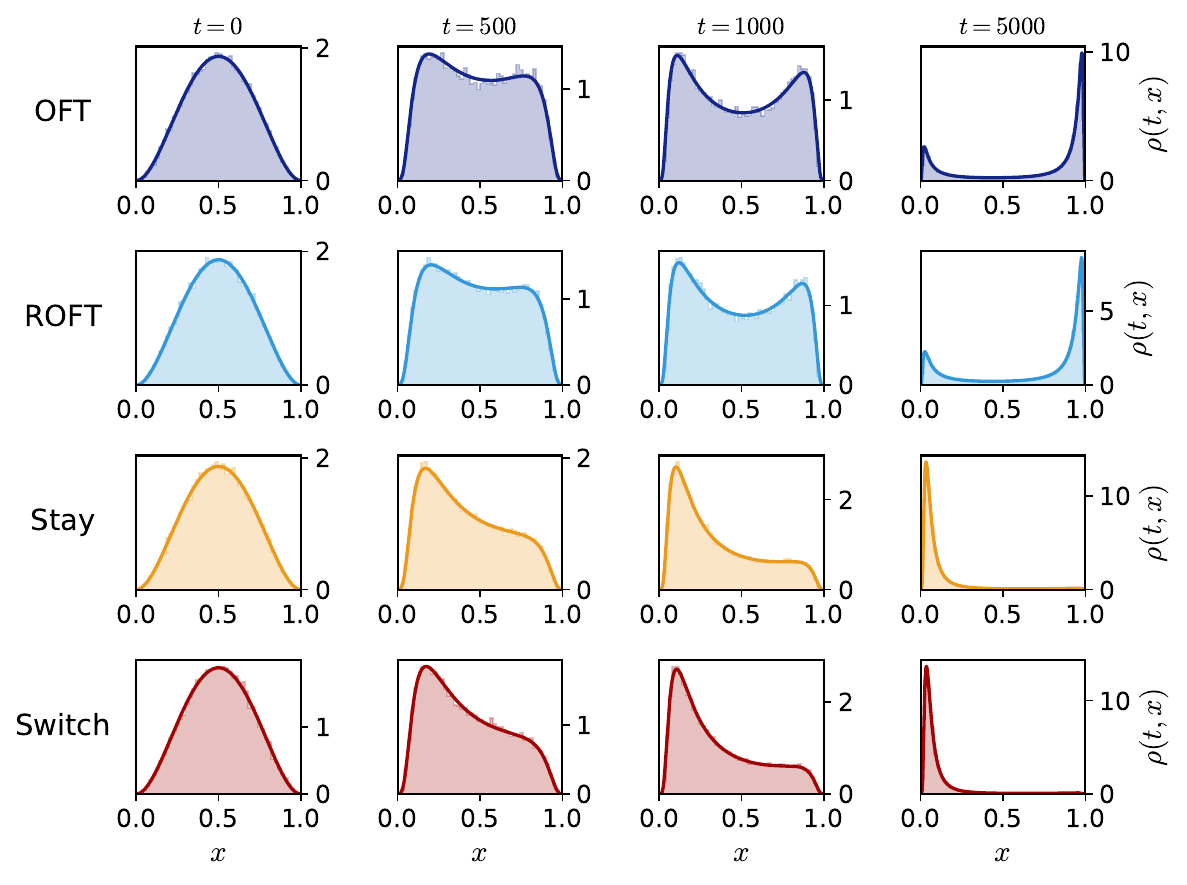}
     \caption{Evolution of the strategy distribution where the population is initialised with $X\sim$ Beta(3,3) for the 4 rules. The theoretical solution (solid line) matches the simulations (histogram).}
     \label{fig: comparison of evolutions for 4 PS rules, uniform}
\end{figure}

\end{document}